\pdfoutput=1
\documentclass[12pt]{article}

\usepackage[utf8]{inputenc} 
\usepackage[T1]{fontenc}    

\usepackage[bitstream-charter]{mathdesign}
\usepackage{amsmath}
\usepackage[scaled=0.92]{PTSans}

\usepackage[
  paper  = letterpaper,
  left   = 1.0in,
  right  = 1.0in,
  top    = 1.0in,
  bottom = 1.0in,
  ]{geometry}

\usepackage[usenames,dvipsnames,table]{xcolor}
\definecolor{shadecolor}{gray}{0.9}

\usepackage[final,expansion=alltext]{microtype}
\usepackage[english]{babel}
\usepackage[parfill]{parskip}
\usepackage{afterpage}
\usepackage{framed}
\newcommand{\spacingset}[1]{%
  \renewcommand{\baselinestretch}{#1}\small\normalsize}
\spacingset{1.2}

{\endMakeFramed}

\usepackage{lineno}

\usepackage{ragged2e}

\newcounter{parcount}

\usepackage{graphicx}
\usepackage{wrapfig}
\usepackage[labelfont=bf,font={footnotesize,stretch=1},width=.9\textwidth]{caption}
\usepackage[format=hang]{subcaption}

\usepackage{booktabs,multirow,multicol}       

\usepackage[algoruled]{algorithm2e}
\usepackage{listings}
\usepackage{fancyvrb}
\fvset{fontsize=\normalsize}

\usepackage[colorlinks,linktoc=all]{hyperref}
\usepackage[all]{hypcap}
\hypersetup{citecolor=MidnightBlue}
\hypersetup{linkcolor=MidnightBlue}
\hypersetup{urlcolor=MidnightBlue}

\usepackage[nameinlink,capitalise]{cleveref}
\Crefname{section}{\S}{\S}
\Crefname{section}{Section}{Sections}

\usepackage[acronym,nowarn]{glossaries}

\lstdefinestyle{mystyle}{
    commentstyle=\color{OliveGreen},
    keywordstyle=\color{BurntOrange},
    numberstyle=\tiny\color{black!60},
    stringstyle=\color{MidnightBlue},
    basicstyle=\ttfamily,
    breakatwhitespace=false,
    breaklines=true,
    captionpos=b,
    keepspaces=true,
    numbers=left,
    numbersep=5pt,
    showspaces=false,
    showstringspaces=false,
    showtabs=false,
    tabsize=2
}
\usepackage{centernot}
\usepackage{amsthm}         
\usepackage{nicefrac}       
\usepackage{mathtools}      
\usepackage{amsbsy}         
\usepackage{amstext}        
\usepackage{thmtools}       
\usepackage{thm-restate}    

\begingroup
    \makeatletter
    \@for\theoremstyle:=definition,remark,plain\do{%
        \expandafter\g@addto@macro\csname th@\theoremstyle\endcsname{%
            \addtolength\thm@preskip\parskip
            }%
        }
\endgroup

\DeclareRobustCommand{\mb}[1]{\ensuremath{\mathbf{\boldsymbol{#1}}}}

\crefname{lemma}{lemma}{lemmas}
\Crefname{lemma}{Lemma}{Lemmas}
\crefname{thm}{theorem}{theorems}
\Crefname{thm}{Theorem}{Theorems}
\crefname{prop}{proposition}{propositions}
\Crefname{prop}{Proposition}{Propositions}
\crefname{assumption}{assumption}{assumptions}
\crefname{assumption}{Assumption}{Assumptions}
\crefname{condition}{condition}{conditions}
\Crefname{condition}{Condition}{Conditions}

\usepackage{arydshln}
\makeatletter
\def\adl@drawiv#1#2#3{%
        \hskip.5\tabcolsep
        \xleaders#3{#2.5\@tempdimb #1{1}#2.5\@tempdimb}%
                #2\z@ plus1fil minus1fil\relax
        \hskip.5\tabcolsep}
\newcommand{\cdashlinelr}[1]{%
  \noalign{\vskip\aboverulesep
           \global\let\@dashdrawstore\adl@draw
           \global\let\adl@draw\adl@drawiv}
  \cdashline{#1}
  \noalign{\global\let\adl@draw\@dashdrawstore
           \vskip\belowrulesep}}
\makeatother

\renewcommand{\epsilon}{\varepsilon}

\declaretheorem[style=plain,numberwithin=section,name=Theorem]{theorem}
\declaretheorem[style=plain,sibling=theorem,name=Lemma]{lemma}
\declaretheorem[style=plain,sibling=theorem,name=Proposition]{proposition}
\declaretheorem[style=plain,sibling=theorem,name=Corollary]{cor}

\declaretheorem[style=definition,sibling=theorem,name=Example]{example}

\declaretheorem[style=plain,sibling=theorem,name=Condition]{condition}

\newenvironment{example*}
 {\pushQED{\qed}\example}
 {\popQED\endexample}
\numberwithin{equation}{section}

\newcommand{\E}[2]{\mathbb{E}_{#1}\left[#2\right]}
\newcommand{\G}[2]{\mathbb{G}_{#1}\left[#2\right]}

\newcommand{\defeq}{\overset{\mathrm{def}}{=}}

\newcommand{\china}[1]{\textcolor{red}{#1}}
\newcommand{\us}[1]{\textcolor{blue}{#1}}
\newcommand{\eu}[1]{\textcolor{green!60!black}{#1}}

\definecolor{WowColor}{rgb}{.75,0,.75}
\definecolor{SubtleColor}{rgb}{0,0,.50}

\newcounter{margincounter}

\usepackage{xr-hyper}
\usepackage{makecell}

\AtBeginEnvironment{abstract}{%
  \renewcommand{\baselinestretch}{1.2}\selectfont}
\usepackage[authoryear,square,semicolon]{natbib}
\usepackage[affil-it]{authblk}
\newcommand{\blind}{0}
\title{Bayesian Poisson-Randomized Gamma Tensor Factorization with Application to \\International Trade Flows}
\date{}

\if0\blind
 \author[2]{Jie Jian}
\author[1,2]{Aaron Schein}
\affil[1]{Department of Statistics, The University of Chicago}
\affil[2]{Data Science Institute, The University of Chicago}
\fi

\begin{document}
\maketitle
\begin{abstract}
    We study sparse semi-continuous tensor data with excess zeros, heavy right tails, and slice-specific dispersion. Such features arise naturally in monetary-valued multi-way data, such as international trade, where most exporter--importer--product--year cells are zero while positive values are continuous and highly variable. To model these data, we propose a Bayesian hierarchical tensor factorization model that places a low-rank CP structure on a latent Poisson rate tensor and couples it with a conditional Gamma model for positive outcomes, with rate parameters that can vary across slices within a mode. The model therefore separates the occurrence and magnitude of positive observations while borrowing strength across all tensor dimensions through a shared low-rank latent structure. To scale posterior inference to large arrays, we develop a hybrid variational--Monte Carlo algorithm that combines efficient coordinate ascent updates with a partially collapsed augmented-data sampler. Applied to approximately 60 million trade flows, the method surfaces multiway dependence across exporters, importers, products, and years that is difficult to recover from gravity-type or pairwise network analyses, which do not jointly model the product and temporal dimensions.
\end{abstract}
\pagebreak

\section{Introduction}

Dynamic multiway networks record pairwise interactions together with additional indexing dimensions, such as relation type, category, or time \citep{kivela2014multilayer}. In many empirical settings, these interactions are observed as nonnegative edge weights on a sparse set of node pairs and indices, yielding a natural tensor representation \citep{DeDomenicoSoleRibaltaCozzoKivelaMorenoPorterGomezArenas2013}. Such weights may be discrete counts, as in social-science and biological networks \citep{schein2015bayesian,ChacoffResascoVazquez2018}, or continuous, as in economic and financial tensor time series \citep{ChenYangZhang2022TensorTS,BillioCasarinIacopiniKaufmann2023BDTR,BabiiGhyselsPan2025TensorPCA}. 
The low-rank structure of such tensors encodes cross-mode dependence, and tensor decomposition has consequently become a standard tool for their analysis \citep{KoldaBader2009, gillis2020nonnegative}.
Existing probabilistic formulations for tensor-valued data have focused mainly on discrete observations, particularly categorical and count-valued data \citep{DunsonXing2009,chi2012tensors, schein2016bayesian}, whereas financial and economic networks often involve monetary-valued edge weights, yielding tensors with excess zeros, heavy-tailed positive entries, and product- or transaction-type-specific variation in scale and dispersion \citep{helpman2008estimating,barigozzi2010multinetwork}. 


Our motivating application is international trade, where modern databases record dollar-valued flows indexed by exporting country $i$, importing country $j$, product $a$, and year $t$, denoted $y_{ijat}\in \mathbb{R}_{\geq 0}$. Countries export distinctive product baskets to particular sets of partners, and shocks propagate through shared supply chains, transport corridors, and policy regimes, generating clustering and co-movement that spans all four modes \citep{Hidalgo2007ProductSpace}. 
The resulting tensor exhibits precisely the challenges described above: for any given product, most country pairs record no trade at all, while active pairs exchange flows ranging from hundreds to billions of dollars. Dispersion also differs markedly across products \citep{fally2018commodity}. For example, petroleum oil is associated with much larger and more volatile flows than cocoa products.

This paper develops a probabilistic tensor factorization built on the nonnegative CANDECOMP/PARAFAC (CP) decomposition \citep{lee1999learning,KoldaBader2009}, which represents the four-way trade tensor through the low-rank signal $y_{ijat} \approx \sum_{k=1}^{K} \theta^{(1)}_{ik}\theta^{(2)}_{jk}\theta^{(3)}_{ak}\theta^{(4)}_{tk}$ via exporter, importer, product, and time ``factors'', that are all non-negative thus respecting the support of the non-negative data. As in many applications, the non-negativity of the factors yields a ``parts-based'' representation that is highly interpretable \citep{gillis2020nonnegative}. \Cref{fig:app_russia} illustrates such an  interpretable structure that is surfaced by our proposed model.


Standard non-negative CP models typically assume conditionally Gaussian or Poisson entries \citep{shashua2005nonnegative, cemgil2009bayesian,gillis2020nonnegative} and govern only the location parameter, like the mean, through the low-rank signal, leaving two limitations. First, they do not accommodate the mixed zero-positive support of semi-continuous data. Second, they offer no mechanism for capturing heterogeneous variability across slices when observations share similar means. We address both through a compound Poisson--Gamma construction, which is conceptually related to two-part ``hurdle'' models \citep{Mullahy1986} and closely connected to Tweedie compound Poisson models \citep{SmythJorgensen2002,DunnSmyth2005}. Each tensor entry is driven by a latent Poisson count whose mean follows the CP decomposition, so zeros arise naturally from the probability of no events. Conditional on a positive count, the observed magnitude follows a Gamma distribution whose shape depends on the latent count. This construction can be viewed as closely related to a Tweedie model with a restricted power parameterization. Crucially, dispersion is not tied to the mean but is instead allowed to vary freely across modes, allowing the heterogeneous variability observed across product slices and other tensor margins to be captured through dedicated mode-specific parameters. Such separate dispersion modeling is new in probabilistic CP-type tensor models.

\begin{figure}[t]
\centering
 \includegraphics[width=\linewidth]{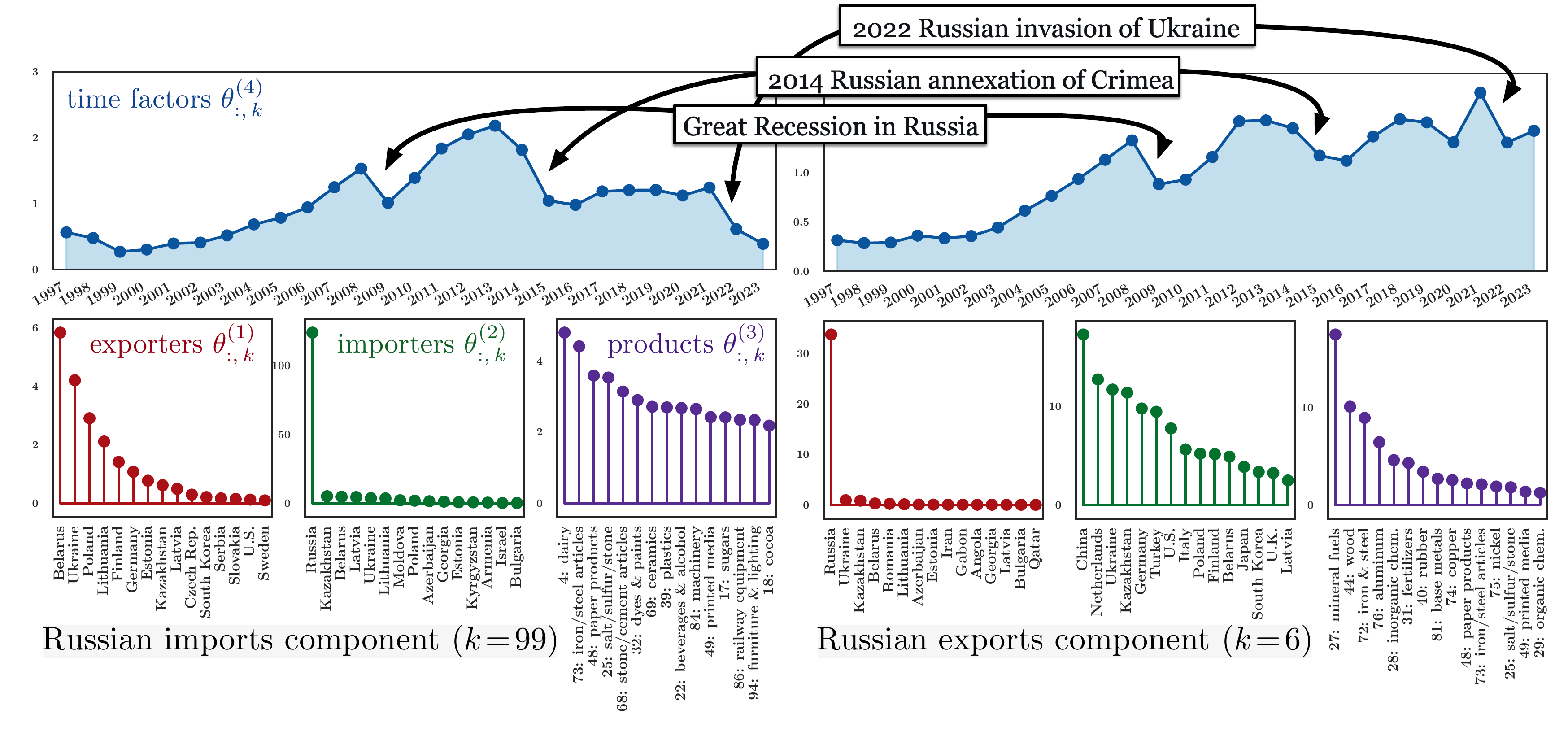}
\caption{\textbf{Russia trading components.}
\textbf{Left:} Russia acts as an importer and manufactures goods sourced primarily from nearby countries. \textbf{Right:} Russia acts as an exporter of mineral, metal, and industrial products to diverse destinations. Both components' time factors show deviations around the 2009 Great Recession, the 2014 annexation of Crimea, and the 2022 invasion of Ukraine.}
  \label{fig:app_russia}
\end{figure}

To enable posterior inference, we introduce two augmentations of the generative process that express the observed dollar-value as a sum of per-component dollar-values. Since the latent Poisson count driving each component can itself be zero, some of these per-component values are exactly zero under the model. To connect the two augmentations and update the per-component values in closed form, we introduce binary indicator variables that explicitly identify active non-zero components, and update them through a ``partially collapsed''~\citep{vanDykPark2008} sampler that marginalizes over variables with which the indicators have a deterministic relationship. Conditioning on these indicators, we show that all remaining parameter blocks can be updated in closed-form via coordinate ascent variational inference (CAVI) \citep{blei2017variational}. The resulting approximate inference algorithm is thus a hybrid of variational and Monte Carlo methods.~\looseness=-1


We apply the proposed model to a large international trade tensor constructed from a World Bank database \citep{wits2025} of size \(151 \times 151 \times 96 \times 27\), recording bilateral trade values among 151 countries as exporters and importers, across 96 high-level product categories, from 1997 to 2023. The model recovers 125 interpretable components, each identifying a coherent combination of exporters, importers, products, and temporal trends. The recovered components reveal economically meaningful structure: concentration and specialization patterns across country and product modes, opposing temporal trends that reflect market-share reallocation across competing exporters, trade-routing shifts around policy changes, and value-chain reconfiguration within industries such as textiles.

Existing statistical approaches to trade data typically operate on fewer dimensions. Gravity models regress bilateral volume on dyad-specific covariates \citep{anderson2011gravity}, while network and community detection methods summarize pairwise trade relationships without incorporating the product or temporal dimensions \citep{wang2023fast, jian2025restricted}. Among tensor-based approaches, \citet{BillioCasarinIacopiniKaufmann2023BDTR} develop a Bayesian dynamic tensor autoregressive model for a two-layer international network of bilateral trade and capital, though at a much smaller scale without granular product-level flows, and \citet{kim2020measuring} use product-level bilateral trade data to cluster exporter-importer pairs by the composition of goods exchanged. Our framework operates on the full four-way structure of transaction values.


After reviewing basic tensor notation and Bayesian Poisson tensor factorization in~\cref{sec:preliminary}, the paper makes four main contributions.

\begin{itemize}
    \item \textbf{Modeling semi-continuous tensor data.}
    In~\cref{method}, we develop a Bayesian hierarchical low-rank tensor model for zero-inflated non-negative continuous arrays, combining a Poisson CP latent-rate structure with a Gamma layer for positive magnitudes.

    \item \textbf{Heterogeneous dispersion.}
    Also in~\cref{method}, 
    we allow the Gamma rate to vary across index-specific subtensors to capture the heterogeneous dispersion.

    \item \textbf{Scalable computation.}
In~\cref{sec:computation}, we develop a hybrid CAVI-sampling algorithm in which conditioning on binary auxiliary variables leads to the optimal variational family involving the lesser-known Bessel distribution~\citep{yuan2000bessel}, and moments that can be computed entirely in closed form, facilitating fast CAVI updates. An asymptotic approximation to the binary auxiliary variables' conditional probabilities further simplifies computation, enabling GPU-accelerated inference for large-scale tensors.
    

    \item \textbf{Empirical validation and international trade application.}
In \cref{sec:simulation}, we present synthetic experiments, and \cref{subsec:trading-oos} evaluates out-of-sample performance on trading data, jointly validating the method under a strong-generalization setting. In \cref{subsec:data}, we describe the trading data and missingness structure, and \cref{subsec:eda} uncovers interpretable multi-way dependence patterns across countries, products, and time using the full trading dataset, which are difficult to detect with standard bilateral summaries. 
\end{itemize}

\section{Preliminaries}
\label{sec:preliminary}

In this section, we briefly review the foundational concepts of non-negative tensor factorization and the Poisson Bayesian tensor factorization approach. For clarity, \cref{sec:preliminary,sec:computation} develop the methodology and inference under a general $M$-way tensor formulation. 

\subsection{CP representation and nonnegative tensor decomposition}

Let $\mathcal{Y}\in \mathbb{R}_{\ge 0}^{I_1 \times I_2 \times \cdots \times I_M}$ be a nonnegative $M$-way tensor with entries $y_{\mathbf{i}}\ge 0$, where $\mathbf{i}\equiv(i_1, \ldots,i_M)$ denotes a multi-index and $i_m\in\{1,\ldots,I_m\}$ for $m=1,\ldots,M$.
We write $\circ$ for the vector outer product. For vectors $\mathbf{a}\in\mathbb{R}^{I_1}$ and $\mathbf{b}\in\mathbb{R}^{I_2}$, 
$\mathbf{a}\circ \mathbf{b}\in\mathbb{R}^{I_1\times I_2}$ yields a matrix with entries
$(\mathbf{a}\circ\mathbf{b})_{ij}=a_i b_j$. This is a rank-one matrix in the usual matrix rank sense.
A rank-$K$ CANDECOMP/PARAFAC (CP) representation writes
\begin{equation}
\label{eq:cp}
\widehat{\mathcal{Y}}
= \sum_{k=1}^K 
\boldsymbol{\theta}^{(1)}_{:,k}\circ \boldsymbol{\theta}^{(2)}_{:,k}\circ \cdots \circ \boldsymbol{\theta}^{(M)}_{:,k},
\qquad
\boldsymbol{\theta}^{(m)}_{:,k}\in \mathbb{R}_{\ge 0}^{I_m},\qquad m=1,\ldots,M
\end{equation}
where $\boldsymbol{\Theta}^{(m)}=[\boldsymbol{\theta}^{(m)}_{:, 1},\ldots,\boldsymbol{\theta}^{(m)}_{:,k}]
\in\mathbb{R}_{\ge 0}^{I_m\times K}$ is the \emph{mode-$m$} factor matrix. Thus an $M$-way tensor admits exactly $M$ such factor matrices, one for each mode.
Equivalently, entry $\mathbf{i}$ is
\begin{equation}
\label{eq:cp-entry}
\widehat{y}_{\mathbf{i}}
= \sum_{k=1}^K \prod_{m=1}^M \theta^{(m)}_{i_m k}.
\end{equation}

Nonnegative tensor decomposition estimates the $M$ mode-specific factor matrices $\{\Theta^{(m)}\}_m$ so that the low-rank reconstruction $\widehat{\mathcal{Y}}$ in~\cref{eq:cp} approximates $\mathcal{Y}$ well. 
Two broad approaches are common. One treats $\widehat{\mathcal{Y}}$ as a deterministic approximation and solves $\min_{\boldsymbol{\Theta}^{(1)},\ldots,\boldsymbol{\Theta}^{(M)}\ge 0}
\ \ \mathcal{L}\!\left(\mathcal{Y},\widehat{\mathcal{Y}}\right)$
where $\mathcal{L}$ is a loss or divergence (e.g., squared Frobenius loss or generalized KL divergence). The other specifies a probabilistic model for $\mathcal{Y}$ with mean or intensity tensor $\widehat{\mathcal{Y}}$ and learns the factors via maximum likelihood or Bayesian inference. Either way, the CP structure yields a parsimonious yet interpretable representation: an unrestricted tensor over $I_1\times\cdots\times I_M$ cells involves on the order of $\prod_{m=1}^M I_m$ parameters, whereas a rank-$K$ CP parameterization uses only $K\sum_{m=1}^M I_m$ mode-specific loadings.

CP factors are identifiable up to componentwise rescaling and permutation. Specifically, for any component $k$ and constants $c_1,\ldots,c_M>0$ with $\prod_{m=1}^M c_m=1$, replacing $\boldsymbol{\theta}^{(m)}_{:,k}$ by $c_m\,\boldsymbol{\theta}^{(m)}_{:,k}$ leaves the implied reconstruction unchanged. This rescaling does not affect the reconstruction \cref{eq:cp-entry}. We therefore do not impose an explicit identifiability constraint. Instead, we focus on within-component comparisons when interpreting results in application in \cref{sec:application}.

\subsection{Bayesian Poisson tensor factorization and Poisson allocation}

For count-valued $M$-way data, Bayesian Poisson tensor factorization (BPTF)~\citep{cemgil2009bayesian,gopalan2015scalable,schein2015bayesian} first models entries $y_{\mathbf{i}}\in\mathbb{Z}_{\ge 0}$ as conditionally independent Poisson random variables whose rates form a low-rank tensor:
\begin{equation}
\label{eq:mu-cp}
\lambda_{\mathbf{i}} = \sum_{k=1}^K \prod_{m=1}^M \theta^{(m)}_{i_m k}, \qquad y_{\mathbf{i}}\mid \lambda_{\mathbf{i}} \sim \ \mathrm{Pois}(\lambda_{\mathbf{i}}).
\end{equation}
The Bayesian formulation then places priors on the nonnegative factors, commonly Gamma for conjugate conditional updates, $\theta^{(m)}_{i_m k}\sim \mathrm{Gam}(\textsc{shape}=a_m,\textsc{rate}=b_m)$. 

Posterior computation exploits an additive augmentation known as Poisson thinning, also called Poisson allocation. 
Specifically, we decompose each observed count $y_{\mathbf{i}}$ into $K$ latent component counts $y^{(k)}_{\mathbf{i}}$ such that
\begin{equation}
\label{eq:poisson-allocation-prelim}
y^{(k)}_{\mathbf{i}}\ \sim\ \ \mathrm{Pois}\Big(\prod_{m=1}^M \theta^{(m)}_{i_m k}\Big),
\qquad
y_{\mathbf{i}}=\sum_{k=1}^K y^{(k)}_{\mathbf{i}}.
\end{equation}
By Poisson additivity, this preserves the marginal model in \cref{eq:mu-cp}. Moreover, conditioning on the total
$y_{\mathbf{i}}$ yields the equivalent multinomial form
\begin{equation}
\label{eq:multinomial-thinning-prelim}
\big(y^{(1)}_{\mathbf{i}},\ldots,y^{(K)}_{\mathbf{i}}\big)\mid y_{\mathbf{i}}, \{\Theta^{(m)}\}_{m}
\sim \mathrm{Mult}\Big(y_{\mathbf{i}};\ \tau^{(1)}_{\mathbf{i}},\ldots,\tau^{(K)}_{\mathbf{i}}\Big),
\end{equation}
where $\tau^{(k)}_{\mathbf{i}}
=\nicefrac{\prod_{m=1}^M \theta^{(m)}_{i_m k}}{\sum_{\ell=1}^K \prod_{m=1}^M \theta^{(m)}_{i_m \ell}}$.
This augmentation restores conditional conjugacy for the factors. In particular, with $\theta^{(m)}_{i_m k}\sim\text{Gam}(a_m,b_m)$, the complete conditional is
\begin{equation}
\label{eq:gamma_posterior}
\theta^{(m)}_{i_m k}\mid \cdot \ \sim\ \mathrm{Gam}\!\left(
a_m+\sum_{\mathbf{i}_{-m}} y^{(k)}_{\mathbf{i}},
\ \ 
b_m+\sum_{\mathbf{i}_{-m}} \prod_{\ell \neq m} \theta^{(\ell)}_{i_\ell k}
\right),
\end{equation}
where $\mathbf{i}_{-m}\equiv(i_1,\ldots,i_{m-1},i_{m+1},\ldots,i_M)$ and $\sum_{\mathbf{i}_{-m}}$ denotes summation over all indices other than $i_m$.
\Cref{eq:multinomial-thinning-prelim,eq:gamma_posterior} form the backbone of scalable Gibbs or variational inference algorithms for Poisson tensor factorization.

\section{Bayesian Poisson-Randomized Gamma Factorization}
\label{method}


We propose the Poisson-randomized Gamma tensor factorization model, which extends the CP decomposition to semi-continuous arrays with many zeros and positive continuous values. The model preserves the interpretability of low-rank CP structure while addressing two limitations of standard Poisson factorization: the inability to accommodate mixed zero-positive support, and the Poisson mean-variance equality that prevents flexible modeling of over-dispersed, heavy-tailed positive magnitudes. In \cref{subsec:model}, we introduce the generative model, and \cref{subsec:reparameterization} develops the reparameterization that enables efficient posterior inference.

\subsection{The Poisson-Randomized Gamma Model}
\label{subsec:model}

To model positive continuous outcomes with exact zeros, we assume the following
\begin{align}
y_{\mathbf{i}}\mid \eta_{\mathbf{i}} \ \overset{\textrm{ind.}}{\sim} \
\begin{cases}
\delta_0, & \eta_{\mathbf{i}}=0, \\
\mathrm{Gam}(\eta_{\mathbf{i}},\beta_{\mathbf{i}}), & \eta_{\mathbf{i}}>0,
\end{cases}
\label{observation-mix-gamma}
\end{align}
where $\eta_{\mathbf{i}}$ is a latent count variable that is Poisson-distributed as
\begin{align}
\eta_{\mathbf{i}} \mid \lambda_{\mathbf{i}} \overset{\textrm{ind.}}{\sim} \mathrm{Pois}(\lambda_{\mathbf{i}}),
\qquad 
\lambda_{\mathbf{i}}=\sum_{k=1}^K \prod_{m=1}^M \theta^{(m)}_{i_m k} .
\label{ijat-Pois-aggregation}
\end{align}
Here $\delta_0$ denotes a Dirac delta function at zero and the Gamma distribution is parameterized by positive shape $\eta_{\mathbf{i}} > 0 $ and rate $\beta_{\mathbf{i}} > 0$.

Marginalizing over the latent count $\eta_{\mathbf{i}}$ then yields the Poisson-randomized Gamma (PRG) distribution \citep{zhou2016augmentable}, which we denote by $y_{\mathbf{i}} \sim \mathrm{PRG}(\lambda_{\mathbf{i}},\beta_{\mathbf{i}})$.
Thus, the hierarchical construction in \cref{observation-mix-gamma,ijat-Pois-aggregation} defines a PRG tensor model with low-rank location parameter $\lambda_{\mathbf{i}}$ and dispersion parameter $\beta_{\mathbf{i}}$, and falls within the general framework of~\cref{eq:cp}. The Poisson component governs the latent event rate, inducing zeros through the probability that $\eta_{\mathbf{i}}=0$ and setting the overall scale of nonzero observations.


In principle, when repeated tensor observations are available, the rate parameter in \cref{observation-mix-gamma} may vary across all cells, yielding a fully cell specific rate tensor with entries $\beta_{\mathbf{i}}$. In practice, however, we typically observe only a single tensor array. Under a single observation, lower-dimensional parameterizations are required to facilitate stable estimation via statistical pooling. This includes, for example, constant rate $\beta$, mode-specific structures such as $\beta_{i_{m^\star}}$, pairwise structures such as $\beta_{i_j,i_a}$, and other lower-dimensional parameterizations. For example, in our international trade application (\cref{sec:application}), restricting $\beta_{\mathbf{i}}\equiv \beta_{i_{m^\star}}$ allows variance dispersion to vary uniquely across products. To maintain a unified presentation that covers any of these pooling structures, we retain the general notation $\beta_{\mathbf{i}}$ throughout the majority of the paper.

An observed zero $y_{\mathbf{i}}=0$ occurs only when $\eta_{\mathbf{i}}=0$, so during inference, $\eta_{\mathbf{i}}=0$ almost surely for all zero entries. More generally, if an entity along any mode (e.g., a particular $i_m$) is consistently associated with small outcomes, the likelihood pushes its corresponding factor loading $\Theta^{(m)}$ toward zero, shrinking its contribution to $\eta_{\mathbf{i}}$. Since these factors are constrained to be nonnegative, we place independent conditionally conjugate Gamma priors over them:
\begin{align}
    \theta_{i_m k}^{(m)} \sim \mathrm{Gam}(\alpha,\alpha \nu^{(m)}), \qquad m=1,\ldots,M,
    \label{prior}
\end{align}
where 
$\mathbb{E}[\theta_{i_m k}^{(m)}]= \nicefrac{\alpha}{\alpha \nu^{(m)}}= \nicefrac{1}{\nu^{(m)}}$ and $\boldsymbol{\nu} = (\nu^{(1)}, \dots, \nu^{(M)})^T$ collects these mode-specific rates.
In all of our experiments, we fix $\alpha=0.1$, which encourages shrinkage among the latent factors.

Having specified the full hierarchical model, we first establish an entrywise identifiability property of the PRG sampling distribution.

\begin{lemma}[Entrywise identifiability of $\mathrm{PRG}(\lambda,\beta)$]
\label{lem:ident_entrywise_mu_beta}
For $\lambda>0$ and $\beta>0$, let $Y\sim \mathrm{PRG}(\lambda,\beta)$ denote the distribution induced by
\begin{equation}\label{eq:prg}
\eta\mid \lambda \sim \mathrm{Pois}(\lambda), 
\qquad
Y\mid \eta,\beta \sim 
\begin{cases}
\delta_0 & \textrm{if }\eta=0,\\[-0.5em]
\mathrm{Gam}(\eta,\beta) &  \textrm{otherwise}.
\end{cases}
\end{equation}
The family $\{\mathrm{PRG}(\lambda,\beta):\lambda>0,\beta>0\}$ is identifiable: if $\mathrm{PRG}(\lambda,\beta)=\mathrm{PRG}(\lambda',\beta')$ for some $\lambda',\beta'>0$, then $\lambda=\lambda'$ and $\beta=\beta'$. A proof of \Cref{lem:ident_entrywise_mu_beta} is given in Supplementary~\ref{app:PRGidentifiability}. 
\end{lemma}

Let $\boldsymbol{\Lambda}, \mathbf{B} \in \mathbb{R}_+^{I_1\times\cdots\times I_M}$ denote the intensity and rate tensors with entries $\lambda_{\mathbf{i}}$ and $\beta_{\mathbf{i}}$. As an entrywise distribution level result, \Cref{lem:ident_entrywise_mu_beta} uniquely determines each parameter pair $(\lambda_{\mathbf{i}}, \beta_{\mathbf{i}})$ and ensures that no simultaneous scaling can produce the same marginal data distribution. Consequently, due to conditional independence, equality of the joint sampling distribution implies entrywise equality of the full parameter tensors, meaning $\boldsymbol{\Lambda} = \tilde{\boldsymbol{\Lambda}}$ and $\mathbf{B} = \tilde{\mathbf{B}}$. To lift this result from $\boldsymbol{\Lambda}$ to its latent factors, we impose the following standard CP uniqueness condition.



\begin{condition}[Kruskal-type CP uniqueness condition]
\label{cond:cp_unique_mu}
Suppose $M\geq 3$ and the intensity tensor admits the rank-$K$ nonnegative CP representation $\lambda_{\mathbf{i}}=\sum_{k=1}^K \prod_{m=1}^M \theta^{(m)}_{i_mk}$, where $\theta^{(m)}_{i_mk}>0 $. Let $k_m$ denote the Kruskal rank of the factor matrix in mode $m$. Assume the factor matrices satisfy the Kruskal-type condition $\sum_{m=1}^M k_m \ge 2K + (M-1)$. Under this condition, the CP representation of $\lambda_{\mathbf{i}}$ is essentially unique up to permutation and scaling~\citep{kruskal1977threeway,sidiropoulos2000uniqueness}.
\end{condition}


Combining the entrywise PRG identifiability with the CP uniqueness condition yields global structural identifiability of the full rate tensor and the latent CP factors:

\begin{cor}[Global identifiability of rate tensors and CP factors] \label{cor:ident_cp_full}
Consider the observed data model induced by \cref{observation-mix-gamma,ijat-Pois-aggregation}, parameterized by the rate tensor $\mathbf{B}$ and the collection of CP factor matrices $\{\Theta^{(m)}\}_{m}$. Suppose that the intensity tensor $\boldsymbol{\Lambda}$ satisfies Condition~\ref{cond:cp_unique_mu}. If two parameter configurations $\left(\{\Theta^{(m)}\}_{m}, \mathbf{B} \right)$ and $\left( \{\tilde{\Theta}^{(m)}\}_{m}, \tilde{\mathbf{B}}\right)$ induce the same joint distribution of the observed tensor $\mathcal{Y}$, then $\tilde{\mathbf{B}} = \mathbf{B}$. Moreover, the two collections of CP factor matrices agree up to the usual permutation and scaling indeterminacies: there exist a permutation $\pi$ of $\{1,\ldots,K\}$ and positive constants $c_{mk}>0$ satisfying $\prod_{m=1}^M c_{mk}=1$ for each $k=1,\ldots,K,$ such that$$ \tilde{\theta}^{(m)}_{i_mk} = c_{mk}\, \theta^{(m)}_{i_m\pi(k)} \qquad \text{for all } m,\ i_m,\ k . $$\end{cor}

\begin{proof}
Because the tensor entries $y_{\mathbf{i}}$ are conditionally independent given $\boldsymbol{\Lambda}$ and $\mathbf{B}$, the marginal distribution of each entry $y_{\mathbf{i}}$ depends solely on its local parameters, reducing exactly to $\mathrm{PRG}(\lambda_{\mathbf{i}}, \beta_{\mathbf{i}})$. Since equality of the joint distribution of the observed tensor $\mathcal{Y}$ inherently implies equality of all its single-entry marginal distributions, it follows that $\mathrm{PRG}(\lambda_{\mathbf{i}}, \beta_{\mathbf{i}}) = \mathrm{PRG}(\tilde{\lambda}_{\mathbf{i}}, \tilde{\beta}_{\mathbf{i}})$ for every multi-index $\mathbf{i}$. 
By \Cref{lem:ident_entrywise_mu_beta}, each marginal PRG law uniquely determines its parameters, yielding $\lambda_{\mathbf{i}} = \tilde{\lambda}_{\mathbf{i}}$ and $\beta_{\mathbf{i}} = \tilde{\beta}_{\mathbf{i}}$ for all $\mathbf{i}$. Therefore, $\boldsymbol{\Lambda} = \tilde{\boldsymbol{\Lambda}}$ and $\mathbf{B} = \tilde{\mathbf{B}}$. \Cref{cond:cp_unique_mu} then guarantees the essential uniqueness of the rank-$K$ CP representation of $\boldsymbol{\Lambda}$, completing the proof.
\end{proof}

\subsection{Two Thinning Representations} 
\label{subsec:reparameterization}

Our goal in the subsequent \cref{sec:computation} is to perform Bayesian inference for the latent factors $\{\Theta^{(m)}\}_{m}$ under the model in \cref{subsec:model}. Direct posterior computation is intractable because the likelihood couples the factors through a sum--product structure, so the full conditionals are not available in closed form. We therefore use an equivalent auxiliary-variable augmentation that preserves the original marginal model while simplifying posterior computation.
Our construction combines Poisson–multinomial allocation with a new thinning for the Gamma layer in \cref{observation-mix-gamma} that links positive magnitudes to the latent structure.

\textbf{Scheme 1: Poisson thinning.}
The first augmentation is Poisson thinning, which provides a latent-source allocation of the Poisson CP mean across the \(K\) rank-one components. This is the tensor analogue of allocative Poisson factorization: each latent count is decomposed into component-specific sub-counts that record how much of the total intensity is attributed to each CP component. In our setting, the same Poisson--multinomial identity used in Bayesian Poisson tensor factorization reviewed in \cref{sec:preliminary} applies to the latent Poisson layer, except that the total count is now the latent variable \(\eta_{\mathbf{i}}\) rather than an observed count---i.e.,
\[
\eta_{\mathbf{i}}\sim \mathrm{Pois}(\mu_{\mathbf{i}}),
\qquad
\mu_{\mathbf{i}}
=
\sum_{k=1}^K \prod_{m=1}^M \theta^{(m)}_{i_mk},
\]
where the per-component sub-counts are then defined as
\begin{equation}
\label{ijat-Pois-thin}
\eta_{\mathbf{i}}
=
\sum_{k=1}^K \eta^{(k)}_{\mathbf{i}},
\qquad
\eta^{(k)}_{\mathbf{i}}
\sim
\mathrm{Pois}\Big(\prod_{m=1}^M \theta^{(m)}_{i_mk}\Big),
\end{equation}
which preserve the original marginal model by Poisson additivity. Equivalently, conditional on the total \(\eta_{\mathbf{i}}\), the sub-counts follow a multinomial allocation,
\begin{equation}
\label{eq:eta_multinomial_thinning}
\big(\eta^{(1)}_{\mathbf{i}},\ldots,\eta^{(K)}_{\mathbf{i}}\big)
\mid \eta_{\mathbf{i}},\{\Theta^{(m)}\}_{m}
\sim
\mathrm{Mult} \Big(
\eta_{\mathbf{i}};
\tau^{(1)}_{\mathbf{i}},\ldots,\tau^{(K)}_{\mathbf{i}}
\Big) \textrm{ where } \tau^{(k)}_{\mathbf{i}}
=
\frac{\prod_{m=1}^M \theta^{(m)}_{i_mk}}
{\sum_{\ell=1}^K \prod_{m=1}^M \theta^{(m)}_{i_m\ell}}.
\end{equation}

This Poisson--multinomial thinning representation is valuable both statistically and computationally: it yields an interpretable component-wise allocation of latent events, restores the conditional structure needed for efficient factor updates, and implies that whenever \(\eta_{\mathbf{i}}=0\), all sub-counts \(\eta^{(k)}_{\mathbf{i}}\) are automatically zero almost surely, so no component-wise allocation is needed for that entry. Hence, this thinning step scales as \(O(\|\mathcal{Y}\|_0 K)\), where \(\|\mathcal{Y}\|_0\) denotes the number of nonzero tensor entries.

\textbf{Scheme 2: Gamma thinning and binary indicators.} In our Poisson-randomized Gamma specification in~\cref{subsec:model}, the latent count $\eta_{\mathbf{i}}$ plays a dual role: it determines whether $y_{\mathbf{i}}$ is zero and, conditional on being positive, acts as the Gamma shape parameter. Under Scheme~1, we thin $\eta_{\mathbf{i}}$ into component-specific sub-counts $(\eta^{(k)}_{\mathbf{i}})_{k=1}^K$. Exploiting the additivity of the Gamma distribution under a common rate, we introduce a matching decomposition of each observation,~\looseness=-1
\[
y_{\mathbf{i}} \;=\; \sum_{k=1}^K y_{\mathbf{i}}^{(k)},
\]
where, conditional on $\eta^{(k)}_{\mathbf{i}}$, the latent contributions are independent with
\begin{equation}
\label{sub-observation}
y_{\mathbf{i}}^{(k)}\mid \eta^{(k)}_{\mathbf{i}} \ \overset{ind.}{\sim}\
\begin{cases}
\delta_0, & \eta^{(k)}_{\mathbf{i}}=0,\\[2pt]
\mathrm{Gam} \big(\eta^{(k)}_{\mathbf{i}},\beta_{\mathbf{i}}\big), & \eta^{(k)}_{\mathbf{i}}>0.
\end{cases}
\end{equation}
This construction is consistent with the original model: $y_{\mathbf{i}}=0$ occurs if and only if $\eta^{(k)}_{\mathbf{i}}=0$ for all $k$, i.e., if the total count $\eta_{\mathbf{i}}=\sum_{k=1}^K \eta^{(k)}_{\mathbf{i}}$ is zero. Consequently, the probability of a zero observation is preserved under thinning, matching \cref{observation-mix-gamma,ijat-Pois-aggregation}.

This formulation, in turn, makes each nonnegative observation $y_{\mathbf{i}}>0$ now a hybrid sum of components that are partly degenerate at zero and partly Gamma contributions, so the marginal likelihood is not available in a tractable form for direct inference. We introduce auxiliary binary activity indicators
$b^{(k)}_{\mathbf{i}}\in\{0,1\}$ that switch component $k$ on or off.
Specifically, $b^{(k)}_{\mathbf{i}} =\mb{1} \{\eta^{(k)}_{\mathbf{i}}>0 \}=\mb{1} \{y^{(k)}_{\mathbf{i}}>0 \}$. 
The active set is then defined as
$ \Delta_{\mathbf{i}}:=\{ k\in \{ 1,\dots, K\}: b_{\mathbf{i}}^{(k)}=1 \}$.
We defer detailed motivation and discussion of the indicator and activity set to \cref{subsec:indicator}. Conditioning on $\Delta_{\mathbf{i}}$ and other latent variables, each positive $y_{\mathbf{i}}$ is a sum of independent Gamma random variables with the same rate. When the sequence of latent variables $(y_{\mathbf{i}}^{(k)})_{k\in \Delta_{\mathbf{i}}}$ for a given multi-index $\mathbf{i}$ consists of independent Gamma random variables sharing a common rate, their normalized vector follows a Dirichlet distribution:
\begin{align}
\big(\nicefrac{y_{\mathbf{i}}^{(k)}}{y_{\mathbf{i}}}\big)_{k\in \Delta_{\mathbf{i}}} \mid \cdot
\sim
\mathrm{Dir}\!\left( (\eta_{\mathbf{i}}^{(k)})_{k\in \Delta_{\mathbf{i}}} \right).
\label{eq:yk_dir}
\end{align}
This yields a second thinning step, which refer to throughout as Gamma--Dirichlet thinning. 

Although one could in principle work only with the first thinning scheme, it leaves an intractable product-sum term that obstructs closed-form updates. The second thinning scheme and the binary indicators decouple this term and are key to making inference practical, as will become clear in \cref{sec:computation}.


We now present the complete hierarchical model, collecting the observation model, latent structure, and priors into a single specification as follows: 
\begin{align*}
\theta^{(m)}_{i_m k} &\sim \mathrm{Gam}\!\left(\alpha,\alpha\nu^{(m)}\right), \qquad m=1,\ldots,M,
&& \text{(factor prior)}\\
\lambda^{(k)}_{\mathbf{i}} &= \prod_{m=1}^M \theta^{(m)}_{i_m k}, 
\qquad 
\eta^{(k)}_{\mathbf{i}} \sim \mathrm{Pois}\!\left(\lambda^{(k)}_{\mathbf{i}}\right),
&& \text{(rank-1 intensity and latent Poisson count)}\\
b^{(k)}_{\mathbf{i}} &= \mb{1}\!\left\{\eta^{(k)}_{\mathbf{i}}>0\right\},
&& \text{(activation indicator)}\\
y^{(k)}_{\mathbf{i}} \mid \eta^{(k)}_{\mathbf{i}} &\sim
\begin{cases}
\delta_0, & \eta^{(k)}_{\mathbf{i}}=0,\\[2pt]
\mathrm{Gam}\!\left(\eta^{(k)}_{\mathbf{i}},\beta_{\mathbf{i}}\right), & \eta^{(k)}_{\mathbf{i}}>0,
\end{cases}
&& \text{(component contribution)}\\
y_{\mathbf{i}} &= \sum_{k=1}^K y^{(k)}_{\mathbf{i}}.
&& \text{(observed entry)}
\end{align*}
This summary consolidates the reparameterizations and augmentation introduced above and will serve as the basis for posterior inference. Throughout the rest of the paper, we work with this model representation.

\section{Posterior Inference}
\label{sec:computation}

Given an observed tensor $\mathcal{Y}$, posterior inference targets the joint distribution of all unknown quantities under the hierarchical model in \cref{method}. We partition these unknowns into the binary activity indicators $\mathcal{A}\stackrel{\mathrm{def}}{=} \{ b^{(k)}_{\mathbf{i}} \}_{\mathbf{i},\,k}$ and the remaining latent variables $\mathcal{Z}\stackrel{\mathrm{def}}{=} \{y_{\mathbf{i}}^{(k)}\}_{\mathbf{i}, k} \;\cup\; \{\eta_{\mathbf{i}}^{(k)}\}_{\mathbf{i}, k} \;\cup\; \{\Theta^{(m)}\}_{m}$. Our goal is therefore to approximate the posterior $p(\mathcal{A},\mathcal{Z} \mid \mathcal{Y},\mathcal{H})$, where $\mathcal{H}\stackrel{\mathrm{def}}{=}\{\alpha\} \;\cup\; \{\boldsymbol{\nu}\} \;\cup\; \{\beta_{\mathbf{i}}\}_{\mathbf{i}}$ denotes the collection of model hyperparameters. In what follows, we treat $\mathcal{H}$ separately from the posterior updates: $\alpha$ is fixed a priori at $0.1$, while $\boldsymbol{\nu}$ and $\{ \beta_{\mathbf{i}}\}$ are updated by an empirical Bayes step. We therefore suppress its dependence in the notation.

In principle, one could perform full MCMC using the conditionals, but this becomes computationally prohibitive for high-dimensional tensors. We therefore adopt a hybrid strategy that alternates between sampling the binary indicators $\mathcal{A}$ and updating the remaining block $\mathcal{Z}$ by coordinate-ascent variational inference (CAVI) conditional on $\mathcal{A}$. We first describe the variational approximation and the general CAVI framework, then derive the model-specific updates, and finally present the sampling step for $\mathcal{A}$.

\subsection{Coordinate-ascent variational inference (CAVI)}
\label{subsec:CAVI}


Conditional on the current binary indicators $\mathcal{A}$, we approximate the conditional posterior $p(\mathcal{Z}\mid \mathcal{Y},\mathcal{A})$ by a tractable variational distribution $q(\mathcal{Z})$. Variational inference replaces the exact posterior by a simpler family of distributions and chooses the member that is closest in Kullback--Leibler divergence \(\mathrm{KL}\!\left(q(\mathcal{Z})\,\|\,p(\mathcal{Z}\mid \mathcal{Y},\mathcal{A})\right)\), or equivalently maximizes the evidence lower bound (ELBO) \citep{Jordan1999VI}, $\mathrm{ELBO}(q)=\E{q} {\log p(\mathcal{Y},\mathcal{Z}\mid \mathcal{A})}-\E{q}{\log q(\mathcal{Z})}$.
The derivation of the ELBO for our model is given in Supplementary~\ref{appendix:elbo}.

In CAVI, one assumes that $q$ factorizes over blocks of $\mathcal{Z}$, and then updates each block holding the others fixed. Under a mean-field approximation, the optimal factor for a block $\mathcal{Z}_j$ satisfies
\[
q^\star(\mathcal{Z}_j)\ \propto\ \exp\!\left\{
\E{q(\mathcal{Z}_{-j})}{\log p(\mathcal{Z}_j\mid \mathcal{Z}_{-j},\mathcal{Y},\mathcal{A})}
\right\},
\]
where $\mathcal{Z}_{-j}$ denotes all remaining blocks. Thus, each variational update has the same functional form as the corresponding complete conditional, with its parameters replaced by expectations under the current variational distribution \citep{blei2017variational}.

\subsection{Complete conditionals and CAVI updates}
\label{subsec:caviupdate}

We next provide the complete conditionals and corresponding variational updates for the latent variables $\mathcal{Z}$, conditional on the sampled indicators $\mathcal{A}$. We also give the associated empirical Bayes updates for the hyperparameters.

\textbf{(i) Update for Dirichlet mixing proportions $\pi^{(k)}_{\mathbf{i}}$ and Gamma sub-mass $y^{(k)}_{\mathbf{i}}$.}
For an entry with $y_{\mathbf{i}}>0$, define proportions over the active set $\pi^{(k)}_{\mathbf{i}}=\nicefrac{y^{(k)}_{\mathbf{i}}}{y_{\mathbf{i}}}$ for $k\in\Delta_{\mathbf{i}}$, so that $\sum_{k\in\Delta_{\mathbf{i}}}\pi^{(k)}_{\mathbf{i}}=1$ and $y^{(k)}_{\mathbf{i}}=\pi^{(k)}_{\mathbf{i}}\,y_{\mathbf{i}}$. The complete conditional is Dirichlet, as in the Gamma thinning step \cref{eq:yk_dir}, with $(\pi^{(k)}_{\mathbf{i}})_{k\in\Delta_{\mathbf{i}}} \mid \cdot\ \sim\ \mathrm{Dir}\big((\eta^{(k)}_{\mathbf{i}})_{k\in\Delta_{\mathbf{i}}}\big)$.
However, because \((\eta^{(k)}_{\mathbf{i}})_{k\in\Delta_{\mathbf{i}}}\) is itself latent, the variational distribution for \((\pi^{(k)}_{\mathbf{i}})_{k\in\Delta_{\mathbf{i}}}\) is written with free Dirichlet parameters \((\omega^{(k)}_{\mathbf{i}})_{k\in\Delta_{\mathbf{i}}}\) as  $q\big((\pi^{(k)}_{\mathbf{i}})_{k\in\Delta_{\mathbf{i}}}\big)=\mathrm{Dir}\big( (\pi^{(k)}_{\mathbf{i}})_{k\in\Delta_{\mathbf{i}}};(\omega^{(k)}_{\mathbf{i}})_{k\in\Delta_{\mathbf{i}}}\big)$.
The coordinate update is then obtained by matching these variational parameters to the expected conditional parameters, yielding $\omega^{(k)}_{\mathbf{i}}=\E{q}{\eta^{(k)}_{\mathbf{i}}}, \ k\in\Delta_{\mathbf{i}}$:
\[
\E{q}{y^{(k)}_{\mathbf{i}}}
= y_{\mathbf{i}}\,\E{q}{\pi^{(k)}_{\mathbf{i}}}
= y_{\mathbf{i}}\,
\frac{\omega^{(k)}_{\mathbf{i}}}{\sum_{\ell\in\Delta_{\mathbf{i}}}\omega^{(\ell)}_{\mathbf{i}}},
\]
and the geometric expectation used in subsequent derivations is
\[
\G{q}{y^{(k)}_{\mathbf{i}}}
=\exp\!\big(\E{q}{\log y^{(k)}_{\mathbf{i}}}\big)
= y_{\mathbf{i}}\,
\exp\!\left\{\psi(\omega^{(k)}_{\mathbf{i}})
-\psi\!\Big(\sum_{\ell\in\Delta_{\mathbf{i}}}\omega^{(\ell)}_{\mathbf{i}}\Big)\right\}.
\]

\textbf{(ii) Update for Poisson sub-counts $\eta^{(k)}$.}
For each component $k$, the conditional model
$
y^{(k)}_{\mathbf{i}}\mid \eta^{(k)}_{\mathbf{i}}
\sim \mathrm{Gam}(\eta^{(k)}_{\mathbf{i}},\beta_{\mathbf{i}}),
$
so that $y^{(k)}_{\mathbf{i}}>0$ implies $\eta^{(k)}_{\mathbf{i}}\ge 1$ almost surely.
Since the prior is $\eta^{(k)}_{\mathbf{i}}\sim \mathrm{Pois}\big(\lambda^{(k)}_{\mathbf{i}}\big)$ with
$\lambda^{(k)}_{\mathbf{i}}=\prod_{m=1}^M \theta^{(m)}_{i_m k}$ denoted as the $k-$th component CP rate,
conditioning only on the event $y^{(k)}_{\mathbf{i}}>0$ yields the zero-truncated Poisson law,
\[
\eta^{(k)}_{\mathbf{i}}\mid \big(y^{(k)}_{\mathbf{i}}>0\big)
\ \sim\ \mathrm{Pois}_+\!\left(\lambda^{(k)}_{\mathbf{i}}\right).
\]

More importantly for inference, conditioning on positive $y^{(k)}_{\mathbf{i}}$ and $\lambda^{(k)}_{\mathbf{i}}$ leads to a closed-form posterior for the integer shape parameter.
Up to constants not depending on $\eta^{(k)}_{\mathbf{i}}=n\ge 1$,
\[
\mathbb{P}\!\left(\eta^{(k)}_{\mathbf{i}}=n \mid y^{(k)}_{\mathbf{i}},\lambda^{(k)}_{\mathbf{i}},\beta_{\mathbf{i}}\right)
\ \propto\
\frac{\big(\lambda^{(k)}_{\mathbf{i}}\big)^n}{n!}\cdot
\frac{\beta_{\mathbf{i}}^n \big(y^{(k)}_{\mathbf{i}}\big)^{n-1}}{\Gamma(n)}.
\]
The normalizing constant is a modified Bessel function of the first kind, so the complete conditional is a zero-truncated Bessel distribution \citep{yuan2000bessel}:
\begin{equation}
\label{eq:mk_bessel}
\eta^{(k)}_{\mathbf{i}}\mid \cdot \ \sim\
\mathrm{Bessel}_{-1}\!\left(2\sqrt{\beta_{\mathbf{i}}\, y^{(k)}_{\mathbf{i}}\,\lambda^{(k)}_{\mathbf{i}}}\right).
\end{equation}
Following \citet{zhou2016augmentable}, if $N\sim \mathrm{Bessel}_{-1}(\alpha)$, then for $n\in\{1,2,\ldots\}$,
\[
\mathbb{P}(N=n)
=
\frac{(\nicefrac{\alpha}{2})^{2n-1}}{I_{-1}(\alpha)\,n!\,\Gamma(n)},
\qquad
I_{-1}(\alpha)
=
\sum_{n=1}^\infty \frac{(\nicefrac{\alpha}{2})^{2n-1}}{n!\,\Gamma(n)}.
\]
Since the order is fixed at $-1$ here, the conditional law is indexed only by the scalar argument $\alpha$. In this fixed-order form, the Bessel distribution is a one-parameter discrete exponential family, and more broadly the Bessel family is a power-series class of underdispersed count distributions. Accordingly, in the variational approximation we write $q\big(\eta^{(k)}_{\mathbf{i}}\big)=\mathrm{Bessel}_{-1}\!\big(\eta^{(k)}_{\mathbf{i}};\,\zeta^{(k)}_{\mathbf{i}}\big)$, where $\zeta^{(k)}_{\mathbf{i}}$ denotes the corresponding Bessel variational parameter.

For $k\in\Delta_{\mathbf{i}}$, the complete conditional in \cref{eq:mk_bessel} is zero-truncated Bessel. Let
$\lambda^{(k)}_{\mathbf{i}}=\prod_{m=1}^M \theta^{(m)}_{i_mk}$, and denote its geometric expectation by
\[
\G{q}{\lambda^{(k)}_{\mathbf{i}}} =\exp\!\big(\E{q}{\log \lambda^{(k)}_{\mathbf{i}}]}\big) = \prod_{m=1}^M \exp\!\big(\E{q}{\log \theta^{(m)}_{i_m k}}\big)
= \prod_{m=1}^M \G{q}{\theta^{(m)}_{i_m k}}.
\]
The corresponding CAVI update therefore remains in the same family:
\[ q\big(\eta^{(k)}_{\mathbf{i}}\big) = \mathrm{Bessel}_{-1}\!\big(\eta^{(k)}_{\mathbf{i}};\,\zeta^{(k)}_{\mathbf{i}}\big),
\qquad
\zeta^{(k)}_{\mathbf{i}}
=
2\sqrt{\beta_{\mathbf{i}}\,\G{q}{y^{(k)}_{\mathbf{i}}}\,\G{q}{\lambda^{(k)}_{\mathbf{i}}}}.
\]
Under the variational distribution
$q(\eta^{(k)}_{\mathbf{i}})=\mathrm{Bessel}_{-1}(\eta^{(k)}_{\mathbf{i}};\,\zeta^{(k)}_{\mathbf{i}})$, the expectation of $\eta^{(k)}_{\mathbf{i}}$ is
\[
\E{q}{\eta^{(k)}_{\mathbf{i}}}
=
\zeta^{(k)}_{\mathbf{i}}
\cdot
\frac{I_{0}(\zeta^{(k)}_{\mathbf{i}})}{I_{-1}(\zeta^{(k)}_{\mathbf{i}})} .
\]
From this update, the necessity of the second thinning scheme and the binary indicators becomes clear in \cref{subsec:reparameterization}. Under the first thinning scheme alone, the Bessel update depends on a sum-product term inside the square root, $\eta^{(k)}_{\mathbf{i}} \mid \cdot \sim\
\mathrm{Bessel}_{-1}\!\left(
2\sqrt{
\beta_{\mathbf{i}}\, y^{(k)}_{\mathbf{i}}
\sum_{k=1}^K \prod_{m=1}^M \theta^{(m)}_{i_m k}}
\right)$,
so inference requires approximating $\mathbb{G}_{q}\!\left[\sum_{k=1}^K \prod_{m=1}^M \theta^{(m)}_{i_m k}\right]$. In our experiments not reported here, delta-method approximations, Monte Carlo estimates, and arithmetic-mean substitutes were all either inaccurate or computationally expensive. The second thinning scheme decouples this bottleneck and restores closed-form inference.


\textbf{(iii) Update for factor matrices $\Theta$.}
Under the Poisson allocation $\eta^{(k)}_{\mathbf{i}}\sim\text{Pois}(\lambda^{(k)}_{\mathbf{i}})$ and the Gamma prior $\theta^{(m)}_{i_m  k}\sim \text{Gam}(\alpha,\alpha\nu^{(m)})$, the complete conditional of $\theta^{(m)}_{i_m  k}$ is Gamma by Poisson--Gamma conjugacy 
\begin{equation}
\label{eq:gamma_posterior_BGPGTF}
\theta^{(m)}_{i_m  k}\mid \cdot\ \sim\ \text{Gam}\!\left(
\alpha+\sum_{{\mathbf{i}_{-m}}} \eta^{(k)}_{\mathbf{i}},
\ \ 
\alpha \nu^{(m)}+\sum_{\mathbf{i}_{-m}} \prod_{\ell\neq m}\theta^{(\ell)}_{i_\ell  k}
\right).
\end{equation}
Indeed, collecting only the terms that depend on $\theta^{(m)}_{i_m  k}$,
\[
\log p(\theta^{(m)}_{i_m  k}\mid \cdot )
= \Big(\alpha-1+\sum_{\mathbf{i}_{-m}}\eta^{(k)}_{\mathbf{i}}\Big)\log \theta^{(m)}_{i_m  k}
-\Big(\alpha\nu^{(m)}+\sum_{\mathbf{i}_{-m}}\prod_{\ell\neq m}\theta^{(\ell)}_{i_\ell  k}\Big)\theta^{(m)}_{i_m  k}
+\text{const}.
\]
Taking expectations with respect to the variational distribution, the CAVI update for $\theta^{(m)}_{i_m k}$ remains in the Gamma family $q\big(\theta^{(m)}_{i_m k}\big) = \mathrm{Gam}\big(\theta^{(m)}_{i_m k};\gamma^{(m)}_{i_m k},\delta^{(m)}_{i_m k}\big)$,
with parameters updated as
\[
\gamma^{(m)}_{i_m k}
=
\alpha+\sum_{\mathbf{i}_{-m}}\E{q}{\eta^{(k)}_{\mathbf{i}}},
\qquad
\delta^{(m)}_{i_m k}
=
\alpha\nu^{(m)}
+
\sum_{\mathbf{i}_{-m}}\prod_{\ell\neq m}\E{q}{\theta^{(\ell)}_{i_\ell k}},
\]
where the product of expectations follows from the mean-field factorization across modes.
With $\psi(\cdot)$ denotes the digamma function, the arithmetic and geometric expectations are
\[
\E{q}{\theta^{(m)}_{i_m k}}
=
\frac{\gamma^{(m)}_{i_m k}}{\delta^{(m)}_{i_m k}},
\qquad
\G{q}{\theta^{(m)}_{i_m k}}
=
\exp\!\Big(\E{q}{\log \theta^{(m)}_{i_m k}}\Big)
=
\frac{\exp\!\big(\psi(\gamma^{(m)}_{i_m k})\big)}{\delta^{(m)}_{i_m k}}.
\]

\textbf{(iv) Empirical Bayes updates for $\beta_{i_{m^\star}}$ and $\boldsymbol{\nu}$.}
The Gamma rate parameters $\beta_{\mathbf{i}}$ and $\nu^{(m)}$ are not assigned priors, so we update them by empirical Bayes, i.e., by maximizing the ELBO with respect to these parameters while holding the current variational distribution fixed \citep{schein2015bayesian}. 
This yields closed-form updates of moment type:
\begin{align}
\nu^{(m)} \ =\ \Big(\tfrac{1}{I_m K}\sum_{i_m=1}^{I_m}\sum_{k=1}^K \E{q}{\theta^{(m)}_{i_m  k}}\Big)^{-1},
\qquad m=1,\ldots,M. \label{eq:beta_n}
\end{align}

For the Gamma rate parameter, the entrywise update takes the form $\beta_{\mathbf{i}} =\nicefrac{\sum_{k=1}^K \E{q}{\eta^{(k)}_{\mathbf{i}}}} {y_{\mathbf{i}}}$.
More generally, when the rate parameter is constrained to be shared across a collection of tensor entries, the corresponding update is obtained by pooling the numerator and denominator over all entries in that collection. In particular, if the rate parameter is assumed to depend only on a designated mode $m^\star$, then the update becomes
\begin{align}
\beta_{i_{m^\star}}
=
\frac{\sum_{k=1}^K \sum_{i_{-m^\star}} \E{q}{\eta^{(k)}_{\mathbf{i}}}}
{\sum_{i_{-m^\star}} y_{\mathbf{i}}},
\label{eq:beta_a}
\end{align}
where $\sum_{i_{-m^\star}}$ denotes summation over all indices except $i_{m^\star}$.


\subsection{Sampling binary indicators}
\label{subsec:indicator}

The binary indicator $b^{(k)}_{\mathbf{i}}$ is updated alternately with the remaining latent variables and hyperparameters. Throughout this subsection, latent quantities such as $\eta^{(k)}_{\mathbf{i}}$ denote their current variational expectations from \cref{subsec:caviupdate}.

This augmented variable addresses a computational bottleneck at the zero--positive boundary. Under the Poisson-randomized Gamma construction (\cref{subsec:reparameterization}), $y_{\mathbf{i}}^{(k)}=0 \iff \eta^{(k)}_{\mathbf{i}}=0$, coupling their supports and rendering the full conditionals degenerate on the boundary. Coordinate-ascent or Gibbs updates that treat these variables separately inherit near-reducibility at this interface: in MCMC, components exhibit poor mixing across the boundary; in mean-field variational inference, the independent factorization assumption conflicts with the coupled support, promoting excessively sparse solutions.

When $y_{\mathbf{i}}=0$, the update is deterministic: $b^{(k)}_{\mathbf{i}}=0$ for all $k$. When $y_{\mathbf{i}}>0$, we sample $b^{(1)}_{\mathbf{i}},\ldots,b^{(K)}_{\mathbf{i}}$ sequentially, each conditioned on $\{b^{(\ell)}_{\mathbf{i}}\}_{\ell\neq k}$, thereby avoiding marginalization over $2^{K-1}$ activation patterns. Crucially, each $b^{(k)}_{\mathbf{i}}$ is drawn from a collapsed conditional that marginalizes out $(y^{(k)}_{\mathbf{i}},\eta^{(k)}_{\mathbf{i}})$; conditioning on either variable directly would render the update degenerate, reintroducing the boundary stickiness the augmentation is designed to eliminate.

Let $y_{\mathbf{i}}^{(-k)}=\{y^{(\ell)}_{\mathbf{i}}\}_{\ell\neq k}$ and $b_{\mathbf{i}}^{(-k)}=\{b^{(\ell)}_{\mathbf{i}}\}_{\ell\neq k}$ denote the collections excluding component $k$, and define $s_{\mathbf{i}}^{(-k)}\;=\;\sum_{\ell\neq k} b^{(\ell)}_{\mathbf{i}}\,\eta^{(\ell)}_{\mathbf{i}}$.
Then, for $y_{\mathbf{i}}>0$,
\begin{equation}
\label{eq:b-bernoulli-update}
\mathbb{P}\!\left(b^{(k)}_{\mathbf{i}}=0 \mid y_{\mathbf{i}},  s_{\mathbf{i}}^{(-k)}, \lambda^{(k)}_{\mathbf{i}} \right)
=\frac{w_0}{w_0+w_1},
\end{equation}
where the two weights correspond to the two competing explanations of $y_{\mathbf{i}}$:
\begin{align*}
w_0
&=\mathbb{P}\!\left(\eta^{(k)}_{\mathbf{i}}=0\mid \lambda^{(k)}_{\mathbf{i}}\right)\,
\text{Gam} \!\left(y_{\mathbf{i}};\, s_{\mathbf{i}}^{(-k)},\, \beta_{\mathbf{i}}\right)
=\exp\!\left(-\lambda^{(k)}_{\mathbf{i}}\right)\ \text{Gam}\!\left(y_{\mathbf{i}};\, s_{\mathbf{i}}^{(-k)},\, \beta_{\mathbf{i}}\right),\\[4pt]
w_1
&=\mathbb{P}\!\left(\eta^{(k)}_{\mathbf{i}}\ge 1\mid \lambda^{(k)}_{\mathbf{i}}\right)\,
f_{y\mid b^{(k)}=1}\!\left(y_{\mathbf{i}}\mid s_{\mathbf{i}}^{(-k)},\lambda^{(k)}_{\mathbf{i}},\beta_{\mathbf{i}}\right)   =\bigl(1-e^{-\lambda^{(k)}_{\mathbf{i}}}\bigr)\, f_{y\mid b^{(k)}=1}\!\left(y_{\mathbf{i}}\mid s_{\mathbf{i}}^{(-k)},\lambda^{(k)}_{\mathbf{i}},\beta_{\mathbf{i}}\right).
\end{align*}
Here $f_{y\mid b^{(k)}=1}$ is the marginal density of $y_{\mathbf{i}}$ when component $k$ is active, obtained by convolving the PRG term of the active component with the Gamma contribution of the others:
\begin{equation}
\label{eq:fy_active_conv}
f_{y\mid b^{(k)}=1}\!\left(y\mid s_{\mathbf{i}}^{(-k)},\lambda^{(k)}_{\mathbf{i}},\beta_{\mathbf{i}}\right)
=
\int_0^{y} \mathrm{PRG}_{+}\!\left(y-\tilde{y};\,\lambda^{(k)}_{\mathbf{i}},\beta_{\mathbf{i}}\right)\, \text{Gam}\!\left(\tilde{y};\,s_{\mathbf{i}}^{(-k)},\beta_{\mathbf{i}}\right)\,d\tilde{y},
\end{equation}
where $\mathrm{PRG}_{+}$ denotes the zero-truncated Poisson--Gamma mixture density induced by
$\eta^{(k)}\sim \mathrm{Pois}_{+}(\lambda^{(k)}_{\mathbf{i}})$ and $y^{(k)}\mid \eta^{(k)}\sim \mathrm{Gam}(\eta^{(k)},\beta_{\mathbf{i}})$. Its closed form is given in Supplementary~\ref{appendix:RGplus}. 


Together, for $y_{\mathbf{i}}>0$ the collapsed Bernoulli update \cref{eq:b-bernoulli-update} admits the closed form
\begin{equation}
\label{eq:Pb}
\mathbb{P}\!\left(b^{(k)}_{\mathbf{i}}=0 \mid y_{\mathbf{i}}, s_{\mathbf{i}}^{(-k)}, \lambda^{(k)}_{\mathbf{i}},\beta_{\mathbf{i}} \right)
=
\frac{\Big(\sqrt{\beta_{\mathbf{i}}\, y_{\mathbf{i}}\, \lambda^{(k)}_{\mathbf{i}}}\Big)^{\,s_{\mathbf{i}}^{(-k)}-1}}
{\Gamma(s_{\mathbf{i}}^{(-k)})\,
I_{s_{\mathbf{i}}^{(-k)}-1}\!\Big(2\sqrt{\beta_{\mathbf{i}}\, y_{\mathbf{i}}\, \lambda^{(k)}_{\mathbf{i}}}\Big)}.
\end{equation}

\begin{algorithm}[H]
\DontPrintSemicolon
\caption{Update of $b^{(k)}_{\mathbf{i}}$}\label{alg:sample-b}
\uIf{$s_{\mathbf{i}}^{(-k)}=0$}{%
  $b^{(k)}_{\mathbf{i}} \gets 1$\;
}\uElseIf{$s_{\mathbf{i}}^{(-k)}>20$}{%
  $p^{(k)}_{\mathbf{i}} \gets 1-\exp(- \nicefrac{\beta_{\mathbf{i}} y_{\mathbf{i}}\lambda^{(k)}_{\mathbf{i}}}{s_{\mathbf{i}}^{(-k)}})$\;
  Sample $b^{(k)}_{\mathbf{i}} \sim \mathrm{Bernoulli}(p^{(k)}_{\mathbf{i}})$\;
}\Else{%
  Sample $b^{(k)}_{\mathbf{i}}$ from \cref{eq:Pb}\;
}
\end{algorithm}

The Bessel function in~\Cref{eq:Pb} is implemented in libraries like \texttt{scipy.special.ive} \citep{SciPyIveDocs} but can be expensive to compute for large $s_{\mathbf{i}}^{(-k)}$. We thus adopt the hybrid sampling rule in~\Cref{alg:sample-b}.

The approximate branch of \Cref{alg:sample-b} is justified by \Cref{prop1} below, which shows that as $s_{\mathbf{i}}^{(-k)}\to\infty$ the exact conditional converges to:
\[
    b^{(k)}_{\mathbf{i}}=\mathbf{1}(\kappa^{(k)}_{\mathbf{i}} >0),\quad \kappa^{(k)}_{\mathbf{i}} \sim \operatorname{Pois}\!\left( 
    \nicefrac{\beta_{\mathbf{i}} y_{\mathbf{i}} \lambda^{(k)}_{\mathbf{i}}}{s_{\mathbf{i}}^{(-k)} }
     \right),
\]
corresponding to the Bernoulli--Poisson link of \cite{zhou2015infinite}. Although the result is asymptotic, the approximation is accurate for moderate $s_{\mathbf{i}}^{(-k)}$, motivating the threshold $20$ above.

\begin{proposition}[Bernoulli--Poisson approximation]
\label{prop1}
Let $\lambda,s,\beta>0$, and let $Y=Y_1+Y_2$ with
$Y_1\sim\mathrm{Gam}(s,\beta)$, $\eta\sim\mathrm{Pois}(\lambda)$, and
$Y_2\mid\eta\sim\mathrm{Gam}(\eta,\beta)$, interpreted as $\delta_0$ when $\eta=0$.
Set $Z=\mb{1}(Y_2>0)$, and let $\widetilde Z=\mb{1}(X>0)$ with
$X\mid Y=y\sim\mathrm{Pois}\bigl( \nicefrac{\lambda\beta y}{s}\bigr)$. Then:
\[
d_{\mathrm{TV}}\!\left(\mathcal{L}(Z),\,\mathcal{L}(\widetilde Z)\right)\longrightarrow 0
\qquad\text{as }s\to\infty,
\]
so the conditional law of $Z$ is asymptotically equivalent to that of its approximation $\widetilde Z$.
\end{proposition}

Explicit expressions for the two conditional Bernoulli success probabilities, together with the proof, are given in Supplementary~\ref{appendix:proof_proposition1}.
\Cref{alg:sample-b}, together with the CAVI updates for $(y^{(k)},\eta^{(k)},\Theta)$, constitutes the full hybrid inference procedure summarized in \Cref{alg:thinGam}.

\begin{algorithm}[t]
\caption{Hybrid CAVI-MC algorithm for BPRGTF}\label{alg:thinGam}
\DontPrintSemicolon
\SetAlgoNoEnd  
\renewcommand{\baselinestretch}{1}\selectfont
\KwIn{$\mathcal{Y}=\{y_{\mathbf{i}}\}$, rank $K$, hyperparameter $\alpha=0.1$}
\KwOut{$q(\Theta), q(\eta), q(\pi)$ and indicators $b$}
\KwData{Initialize $\{\gamma^{(m)}_{i_m k},\delta^{(m)}_{i_m k}\}$, $\{\omega^{(k)}_{\mathbf{i}}\}$, $\{\zeta^{(k)}_{\mathbf{i}}\}$, $\{b^{(k)}_{\mathbf{i}}\}$}
\While{ELBO not converged \textnormal{(see Supplementary~\ref{appendix:elbo})}}{
  \ForEach{observed tensor entry $\mathbf{i}$ with $y_{\mathbf{i}}>0$}{
    \For(\tcp*[f]{Update Bernoulli $b^{(k)}_{\mathbf{i}}$}){$k \in \{1,\ldots,K\}$}{
      Sample $b^{(k)}_{\mathbf{i}}$ and update $\Delta_{\mathbf{i}}$ via \cref{alg:sample-b}\;
    }
  }
  \ForEach{observed tensor entry $\mathbf{i}$ with $y_{\mathbf{i}}>0$}{
    \For(\tcp*[f]{Update $\mathrm{Dirichlet}((\nicefrac{y^{(k)}_{\mathbf{i}}}{y_{\mathbf{i}}})_{k\in\Delta_{\mathbf{i}}};\,(\omega^{(k)}_{\mathbf{i}})_{k\in\Delta_{\mathbf{i}}})$}){active component $k\in\Delta_{\mathbf{i}}$}{
      $\omega^{(k)}_{\mathbf{i}} \leftarrow \E{q}{\eta^{(k)}_{\mathbf{i}}}$\;
      $\E{}{y^{(k)}_{\mathbf{i}}} \leftarrow y_{\mathbf{i}}\,\tfrac{\omega^{(k)}_{\mathbf{i}}}{\sum_{\ell\in\Delta_{\mathbf{i}}}\!\omega^{(\ell)}_{\mathbf{i}}}$\;
      $\G{}{y^{(k)}_{\mathbf{i}}} \leftarrow y_{\mathbf{i}}\exp\!\big\{\psi(\omega^{(k)}_{\mathbf{i}}) - \psi(\textstyle\sum_{\ell}\omega^{(\ell)}_{\mathbf{i}})\big\}$\;
    }
  }
  \ForEach{observed tensor entry $\mathbf{i}$ with $y_{\mathbf{i}}>0$}{
    \For(\tcp*[f]{Update $\mathrm{Bessel}_{-1}(\eta^{(k)}_{\mathbf{i}};\zeta^{(k)}_{\mathbf{i}})$}){active component $k\in\Delta_{\mathbf{i}}$}{
      $\zeta^{(k)}_{\mathbf{i}} \leftarrow 2\sqrt{\beta_{\mathbf{i}}\,\G{}{y^{(k)}_{\mathbf{i}}}\cdot \prod_{m=1}^M \G{}{\theta^{(m)}_{i_m k}}}$\;
      $\E{q}{\eta^{(k)}_{\mathbf{i}}} \leftarrow \tfrac{1}{2}\zeta^{(k)}_{\mathbf{i}}\, \tfrac{I_0(\zeta^{(k)}_{\mathbf{i}})}{I_{-1}(\zeta^{(k)}_{\mathbf{i}})}$\;
    }
  }
  \For{component $k \in \{1, \dots, K\}$}{
    \For{mode $m \in \{1,\dots, M\}$}{
      \For(\tcp*[f]{Update $\mathrm{Gamma}(\theta^{(m)}_{i_m k};\gamma^{(m)}_{i_m k},\delta^{(m)}_{i_m k})$}){entity $i_m \in \{1,\dots, I_m\}$}{
        $\gamma^{(m)}_{i_m k} \leftarrow \alpha + \sum_{\mathbf{i}_{-m}} \E{q}{\eta^{(k)}_{\mathbf{i}}}$\;
        $\delta^{(m)}_{i_m k} \leftarrow \alpha\nu^{(m)} + \sum_{\mathbf{i}_{-m}} \prod_{\ell\neq m}\E{}{\theta^{(\ell)}_{i_\ell k}}$\;
        $\E{}{\theta^{(m)}_{i_m k}} \leftarrow \frac{\gamma^{(m)}_{i_m k} }{\delta^{(m)}_{i_m k}}$\;
        $\G{}{\theta^{(m)}_{i_m k}} \leftarrow \frac{\exp(\psi(\gamma^{(m)}_{i_m k}))}{\delta^{(m)}_{i_m k}}$\;
      }
    }
  }
}
\end{algorithm}

\section{Synthetic Simulation}
\label{sec:simulation}





In this section, we evaluate the proposed model against recent baselines in a series of controlled synthetic experiments designed to assess predictive performance under the data-generating process described in \cref{method}. The same experimental setup, baselines, and reconstruction method are reused in \cref{subsec:trading-oos} for out-of-sample prediction on the real trade data.

\textbf{Data.} We generate a four-way tensor $\mathcal{Y} \in \mathbb{R}_{\ge 0}^{I \times J \times A \times T}$ with rank $R=20$ and mode dimensions $(I,J,A,T) = (50,40,30,10)$. We first construct four factor matrices $\Theta^{(1)} {\in} \mathbb{R}_{\ge 0}^{I \times R}$, $\Theta^{(2)} {\in} \mathbb{R}_{\ge 0}^{J \times R}$, $\Theta^{(3)} {\in} \mathbb{R}_{\ge 0}^{A \times R}$, and $\Theta^{(4)} \in \mathbb{R}_{\ge 0}^{T \times R}$ by first drawing each entry from $\mathrm{Gam}(1,1)$ and then imposing sparsity by setting each entry to zero with probability $0.3$. We then define $\lambda_{ijat}{=}\sum_{r=1}^R \theta_{ir}^{(1)}\theta_{jr}^{(2)}\theta_{ar}^{(3)}\theta_{tr}^{(4)}$, sample $\eta_{ijat}\sim\mathrm{Pois}(\lambda_{ijat})$, and set $y_{ijat}{=}0$ if $\eta_{ijat}{=}0$ and $y_{ijat}{\sim}\mathrm{Gam}(\eta_{ijat},\beta_a)$ otherwise, where $(\beta_a)_{a=1}^A$ is set to an equally spaced grid from $0.01$ to $0.1$. For each tensor, we then generate 30 random 80-20 train-test splits.~\looseness=-1 

\textbf{Models.} We compare against two baselines for nonnegative tensors with zeros and positive continuous values: non-negative tensor factorization minimizing the least squared loss (NTF-LS) and minimizing the $\beta$-divergence with $\beta\!=\!1.5$ (NTF-$\beta$), with the latter corresponding to a Tweedie loss. To fit NTF-LS we use the implementation in \textsc{Tensorly}~\citep{JMLRtensorly} and for NTF-$\beta$ we use \textsc{NNEinFact}~\citep{hood2026near}. We also consider two variants of our method: \textit{shared}, which uses a common Gamma rate $\beta_{\mathbf{i}}=\beta$ across all entries, and \textit{local}, which allows the rate to vary along one designated mode $\beta_{\mathbf{i}}=\beta_a$. Each method is fit to each train-test split of each data tensor with $K\in\{15,20,25\}$.

\textbf{Evaluation.} We evaluate performance via mean absolute error over all held-out entries (MAE) and over nonzero held-out entries (MAE-NZ), mean absolute relative error over nonzero entries (MArel-NZ), Hamming loss on zeros (HAM-Z), and the area under the ROC curve (AUC) for zero versus nonzero entries. For NTF-$\beta$, we use multiple learning rates and report the best in terms of test-set MAE. For NTF-LS and NTF-$\beta$, held-out entries are reconstructed from point-estimated factors, while for BPGPTF they are reconstructed with the geometric mean under the variational posterior, $\widehat{y}_{\mathbf{i}}=\nicefrac{\sum_{k=1}^K \prod_{m=1}^M \G{} {\theta^{(m)}_{i_m k}}}{\beta_{\mathbf{i}}}$.

%
 
\Cref{tab:synthetic_K_20} reports the mean and standard error over 30 random train--test splits, where the proposed model with the local Gamma rate is labeled \textit{local}. The \textit{local} model achieves the best predictive performance across all metrics and rank settings, with particularly large gains in the Hamming loss on zeros, indicating improved recovery of sparsity. For all methods, predictive performance generally improves as $K$ increases.

\begin{table}[!t]
\caption{\textbf{Synthetic-data predictive performance with row-wise optima bolded:} the proposed model with the local Gamma rate achieves the best performance across all metrics and scenarios. Entries are mean (SE) over 30 random 80/20 entrywise train--test splits.}
\label{tab:synthetic_K_20}
\centering
{\renewcommand{\baselinestretch}{1}\selectfont
\setlength{\tabcolsep}{6pt}
\begin{tabular}{@{}clcccc@{}}
\toprule
 & & \multicolumn{2}{c}{Baselines} & \multicolumn{2}{c}{Proposed} \\
\cmidrule(lr){3-4} \cmidrule(lr){5-6}
$K$ & Metric & NTF-LS & NTF-$\beta$ & shared $(\beta_{\mathbf{i}}=\beta)$ & local $(\beta_{\mathbf{i}}=\beta_a)$ \\
\midrule
\multirow{5}{*}{15}
& MAE ($\downarrow$)      & 20.6 (0.5)  & 28.4 (0.8)  & 22.3 (0.7)  & \textbf{18.2 (0.5)}  \\
& MAE-NZ ($\downarrow$)   & 56.3 (0.9)  & 67.3 (1.2)  & 57.5 (1.0)  & \textbf{55.3 (0.8)}  \\
& MArel-NZ ($\downarrow$) & 4.6 (0.3)   & 5.5 (0.3)   & 5.7 (0.3)   & \textbf{3.2 (0.1)}   \\
& HAM-Z ($\downarrow$)    & 0.70 (0.01) & 0.89 (0.01) & 0.95 (0.01) & \textbf{0.43 (0.01)} \\
& AUC ($\uparrow$)        & 0.86 (0.01) & 0.82 (0.00) & 0.89 (0.00) & \textbf{0.91 (0.00)} \\
\midrule
\multirow{5}{*}{20}
& MAE ($\downarrow$)      & 19.8 (0.4)  & 26.3 (0.8)  & 20.1 (0.5)  & \textbf{16.8 (0.4)}  \\
& MAE-NZ ($\downarrow$)   & 54.3 (0.7)  & 63.8 (1.2)  & 53.7 (0.6)  & \textbf{51.6 (0.6)}  \\
& MArel-NZ ($\downarrow$) & 4.7 (0.3)   & 5.2 (0.2)   & 5.4 (0.3)   & \textbf{3.2 (0.2)}   \\
& HAM-Z ($\downarrow$)    & 0.67 (0.01) & 0.85 (0.01) & 0.89 (0.01) & \textbf{0.34 (0.01)} \\
& AUC ($\uparrow$)        & 0.88 (0.00) & 0.84 (0.01) & 0.92 (0.00) & \textbf{0.94 (0.00)} \\
\midrule
\multirow{5}{*}{25}
& MAE ($\downarrow$)      & 19.4 (0.4)  & 24.3 (0.7)  & 19.3 (0.5)  & \textbf{16.6 (0.4)}  \\
& MAE-NZ ($\downarrow$)   & 53.5 (0.7)  & 60.8 (1.1)  & 52.4 (0.6)  & \textbf{51.2 (0.5)}  \\
& MArel-NZ ($\downarrow$) & 4.8 (0.3)   & 5.2 (0.3)   & 5.3 (0.3)   & \textbf{3.3 (0.2)}   \\
& HAM-Z ($\downarrow$)    & 0.65 (0.01) & 0.81 (0.01) & 0.84 (0.01) & \textbf{0.32 (0.01)} \\
& AUC ($\uparrow$)        & 0.90 (0.00) & 0.86 (0.00) & 0.93 (0.00) & \textbf{0.94 (0.00)} \\
\bottomrule
\end{tabular}
}
\end{table}

\section{Application to International Trade Data}
\label{sec:application}

We apply the proposed model to international merchandise trade data. In \cref{subsec:data} we describe the dataset and preprocessing. We then evaluate out-of-sample predictive performance under the strong-generalization regime in \cref{subsec:trading-oos}. In \cref{subsec:eda}, we illustrate exploratory analysis through a case study of the fitted components.

\subsection{Data and preprocessing}
\label{subsec:data}

We analyze annual merchandise-trade flows from the World Integrated Trade Solution (WITS) \citep{wits2025}, a World Bank platform providing standardized access to trade data reported by national customs authorities. We use export values at the HS-1996 2-digit (HS2) level, yielding all 96 product categories (e.g., 01 Live animals, 84 Machinery, 85 Electrical machinery) for 196 countries from 1997 to 2023, measured in thousand USD. We arrange the data as a four-mode tensor $\mathcal{Y}\in\mathbb{R}_{\ge0}^{I\times J\times A\times T}$, where $i$ indexes exporters, $j$ indexes importers, $a$ indexes HS2 products, and $t$ indexes years. Each entry $y_{ijat}$ records the value of product $a$ exported from country $i$ to country $j$ in year $t$, so $y_{ijat}=0$ may reflect either a true zero flow or an unreported observation. More details are in Supplementary~\ref{app:preprocessing}.

Missingness is pervasive in international trade data, arising from reporter-year non-submission, confidentiality suppressions, reporting thresholds, and HS revisions \citep{UNComtradeMethodology,WITSUserGuide}. We preprocess the data by flagging reporter-year lapses (exporters with no positive flows in a given year), imputing short internal gaps, and retaining a binary observation mask so that structurally missing entries do not contribute to estimation; details are in Supplementary~\ref{app:masking}. The resulting tensor covers 151 countries, 27 years, 96 products, and 59,100,192 entries, of which 21.3\% are nonzero. To allow product-specific dispersion, we set the Gamma rate to vary across HS2 categories: $y_{ijat}\mid \eta_{ijat} \overset{ind.}{\sim} \mathrm{Gam}(\eta_{ijat},\beta_a)$.

\subsection{Out-of-sample predictive performance}
\label{subsec:trading-oos}

We adopt the strong-generalization regime of \cite{marlin2004collaborative}, using a leave-one-year-out design. For each year $t_0 \in \{1,\dots,T\}$, we fit the model to the tensor $\mathcal{Y}_{\neg t_0} \in \mathbb{R}_{\geq 0}^{I \times J \times A \times (T-1)}$ defined as the concatenation $\mathcal{Y}_{\neg t_0} \defeq \textrm{concat}(\{\mathcal{Y}_{:,:,:,t}\}_{t \neq t_0})$ of the set of all 3-mode slices for each year $t \neq t_0$. We then split $\mathcal{Y}_{t_0} = \mathcal{Y}_{t_0}^{\textrm{obs}} \cup  \mathcal{Y}_{t_0}^{\textrm{eval}}$ into a 75\% observed set and a 25\% evaluation set, refitting only the time-mode factor $\Theta^{(4)}$ on the observed set while fixing $\Theta^{(1)}$, $\Theta^{(2)}$, and $\Theta^{(3)}$. We consider three masking schemes: random holdout, complement-of-dense-block, and dense-block, following the evaluation design of \citet{schein2015bayesian}. The dense block is defined by the $N_0$ highest-volume countries in each held-out year; see Supplementary~\ref{app:masking} for details and an illustration.

\Cref{fig:metrics_over_K} summarizes performance across $K\in\{5,10,50,100,125\}$ for the three masking schemes; full numerical results at $K=125$ are in Supplementary~\ref{app:oos_table}. NTF-LS is competitive at small ranks but saturates as $K$ increases. Our proposed model with the local Gamma rate benefits substantially from additional rank, consistently outperforming all baselines at moderate to large $K$ across masking regimes.

\begin{figure}[ht]
  \centering
  \includegraphics[width=\textwidth]{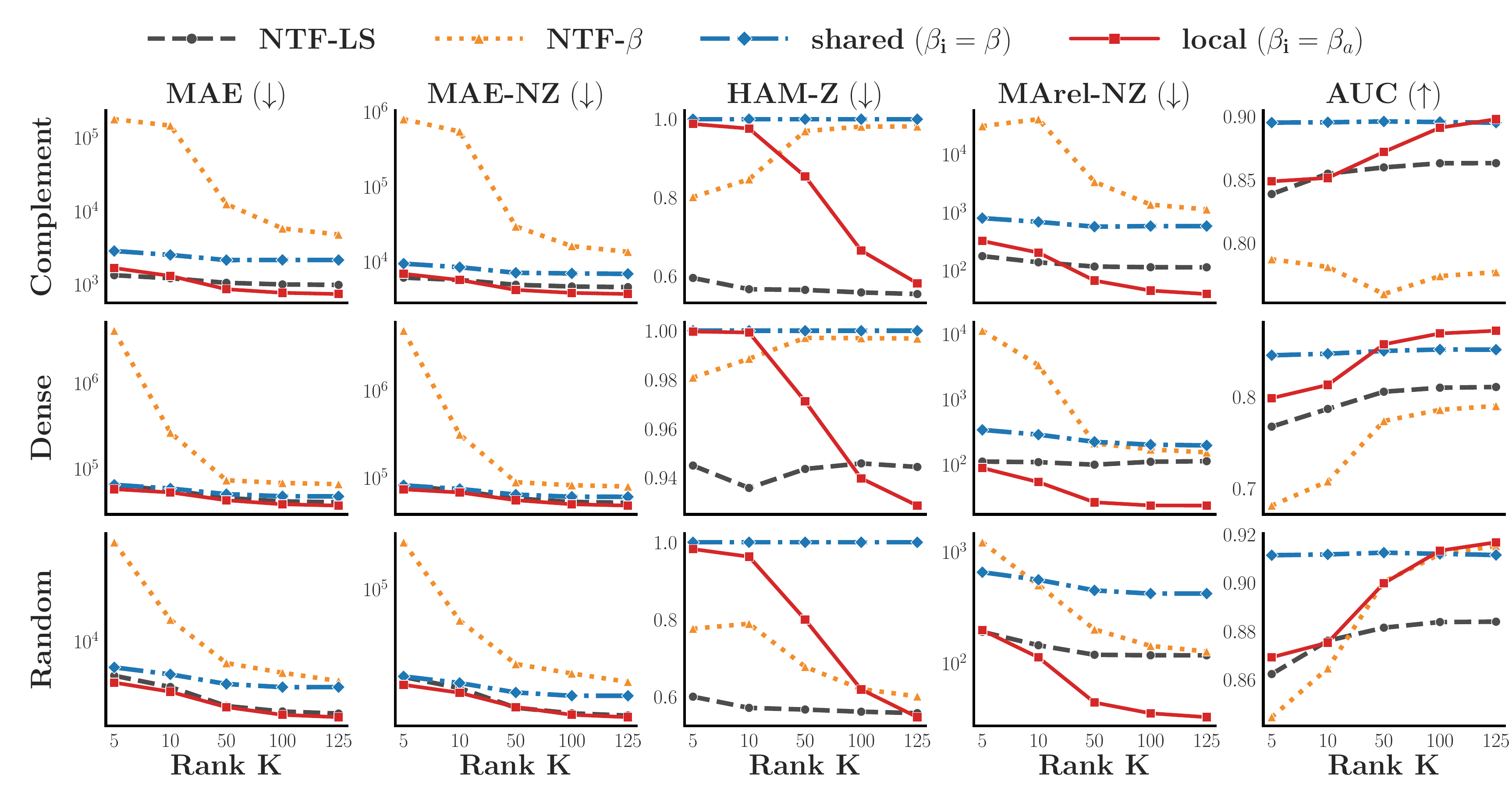}
\caption{\textbf{Real-data predictive performance across ranks and masking schemes:} the proposed model with the local Gamma rate (red) benefits from increasing rank and becomes the top performer on most metrics at moderate to large $K$.}
  \label{fig:metrics_over_K}
\end{figure}

\subsection{Exploratory case study}
\label{subsec:eda}

We now illustrate how the proposed model can be used for exploratory analysis. Since $K=125$ gives the best out-of-sample predictive performance, we refit the model to the full unmasked dataset. The estimated product-specific Gamma rates $\beta_a$ are strongly negatively associated with both the mean and variance of nonzero entries in each product subtensor (Spearman $\rho=-0.991$ and $-0.947$), consistent with heavier industries exhibiting larger and more volatile flows; see Supplementary~\ref{subsec:beta_values} for details.

For each component $k$, the nonnegative factors $\boldsymbol{\theta}_{:,k}^{(1)}$, $\boldsymbol{\theta}_{:,k}^{(2)}$, $\boldsymbol{\theta}_{:,k}^{(3)}$, $\boldsymbol{\theta}_{:,k}^{(4)}$ describe its exporter, importer, product, and temporal structure. Earlier \Cref{fig:app_russia} illustrates a component dominated by Russian exports of mineral, metal, and industrial products, with temporal deviations around the 2009 crisis, 2014 Crimea annexation, 2020 pandemic, and 2022 Russia--Ukraine war. A small number of components reflect reporting conventions or definitional breaks rather than substantive trade structure; see Supplementary~\ref{app:data_artifacts} for an example.

\begin{figure}[ht]
  \centering
  \includegraphics[width=0.9\linewidth]{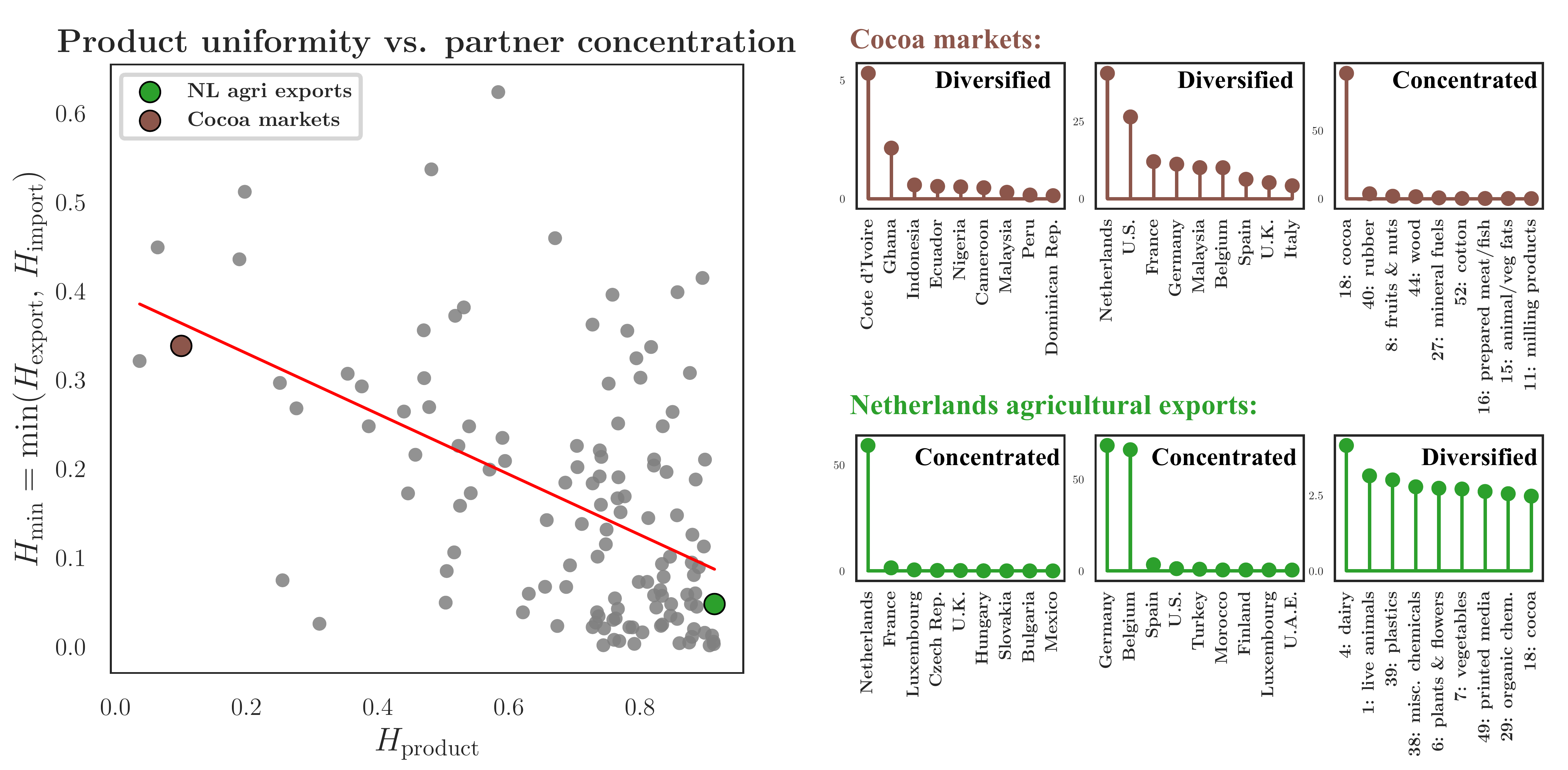}
  \caption{\textbf{Product uniformity vs.\ partner concentration.} \textbf{Left:} Each point is a component, with $x$-coordinate $H_{\text{product}}$ (entropy of the product distribution within the component) and $y$-coordinate $H_{\min} = \min(H_{\text{export}}, H_{\text{import}})$ (the smaller of the exporter and importer entropies). Higher values indicate more uniform composition. The red line confirms that broader product composition coincides with more concentrated partner profiles. \textbf{Right:} Two highlighted examples illustrate the extremes: cocoa markets (brown) concentrate on a few products but reach many partners; Netherlands agricultural exports (green) show the reverse pattern. A binned summary is in Supplementary~\ref{app:entropy_binned}.}
  \label{fig:entropy_example}
\end{figure}

\textbf{Concentration and entropy.} To compare components systematically, we normalize each factor vector into shares ($\tilde{\theta}_{ik} = \nicefrac{\theta^{(1)}_{ik}}{\sum_{i'} \theta^{(1)}_{i'k}}$, analogously for importers and products) and measure each by the normalized Shannon entropy $H_k = -\nicefrac{1}{\ln I} \cdot \sum_{i} \tilde{\theta}_{ik}\ln\tilde{\theta}_{ik} \in [0,1]$, where $I$ is the number of entities in the mode and $H_k=1$ indicates a uniform distribution while $H_k\to 0$ indicates full concentration on a single entity. \Cref{fig:entropy_example} reveals a sharp cross-mode trade-off: components with more uniform product coverage ($H^P_k$ high) systematically have more concentrated partner profiles ($H_{\min,k} = \min\{H^E_k, H^M_k\}$ low). Components anchored to a single product class tend to involve many trading partners, reflecting either the geographic concentration of natural resources (e.g.\ cocoa, coffee) or a small number of technologically advanced exporters (e.g.\ aircraft, pharmaceuticals), while components spanning many products are driven by a few large, diversified economies.

\newlength{\tradepanelheight}
\setlength{\tradepanelheight}{3.9cm} 

\begin{figure}[ht]
    \centering

    \begin{subfigure}[b]{0.50\textwidth}
        \centering
        \vspace{0pt}
        \resizebox{!}{\tradepanelheight}{%
\renewcommand{\baselinestretch}{1}\selectfont  
\tiny

\begin{tabular}{lclc}
\toprule
\textbf{Most symmetric} & GL & \textbf{Most asymmetric} & GL \\
\midrule
\eu{Germany} -- \eu{Luxembourg}   & 0.85 & Chile -- \china{China}               & 0.07 \\
\eu{Austria} -- \eu{Hungary}      & 0.84 & Brazil -- Russia                     & 0.08 \\
\eu{Czech Rep.} -- \eu{Hungary}   & 0.84 & Brazil -- \china{China}              & 0.10 \\
\eu{Czech Rep.} -- \eu{Romania}   & 0.84 & \china{China} -- Peru                & 0.10 \\
\eu{Belgium} -- \eu{Luxembourg}   & 0.83 & Pakistan -- \us{U.S.}                & 0.13 \\
\eu{Austria} -- \eu{Czech Rep.}   & 0.83 & Argentina -- \eu{Spain}              & 0.14 \\
\eu{Czech Rep.} -- \eu{Slovakia}  & 0.82 & Brazil -- Taiwan                     & 0.15 \\
\eu{Hungary} -- \eu{Slovakia}     & 0.82 & El Salvador -- \us{U.S.}             & 0.15 \\
\eu{Austria} -- \eu{Slovenia}     & 0.80 & Chile -- \eu{Spain}                  & 0.15 \\
\eu{Belgium} -- \eu{Netherlands}  & 0.80 & India -- Russia                      & 0.16 \\
\eu{France} -- \eu{Germany}       & 0.79 & \china{China} -- New Zealand         & 0.16 \\
\eu{Germany} -- \eu{U.K.}         & 0.78 & Australia -- Japan                   & 0.16 \\
\eu{Hungary} -- \eu{Romania}      & 0.78 & Argentina -- \china{China}           & 0.17 \\
\eu{Austria} -- \eu{Germany}      & 0.78 & Canada -- India                      & 0.17 \\
\eu{Austria} -- \eu{Slovakia}     & 0.77 & Chile -- \us{U.S.}                   & 0.17 \\
\eu{Belgium} -- \eu{France}       & 0.77 & Brazil -- South Korea                & 0.17 \\
\eu{Czech Rep.} -- \eu{Poland}    & 0.76 & Brazil -- Japan                      & 0.17 \\
\eu{Italy} -- \eu{Slovenia}       & 0.75 & Australia -- \china{China}           & 0.18 \\
\eu{Czech Rep.} -- \eu{Germany}   & 0.75 & Brazil -- Canada                     & 0.18 \\
\eu{Germany} -- \eu{Netherlands}  & 0.75 & \china{China} -- \eu{Norway}         & 0.18 \\
\eu{Austria} -- \eu{Poland}       & 0.75 & Russia -- \eu{Switzerland}           & 0.21 \\
\eu{Germany} -- \eu{Slovenia}     & 0.74 & \eu{Italy} -- \eu{Norway}            & 0.21 \\
South Korea -- Taiwan             & 0.74 & Canada -- \china{China}              & 0.21 \\
\eu{Hungary} -- \eu{Poland}       & 0.73 & \us{U.S.} -- Vietnam                 & 0.21 \\
\eu{Poland} -- \eu{Slovakia}      & 0.73 & \eu{Germany} -- Vietnam              & 0.22 \\
\bottomrule[1pt]
\end{tabular}\vspace{0.3ex}
        }
        \caption{Most/least symmetric trade pairs}
        \label{tab:gl}
    \end{subfigure}
    \hfill
    \begin{subfigure}[b]{0.49\textwidth}
        \centering
        \vspace{0pt}
        \includegraphics[
            height=2\tradepanelheight
        ]{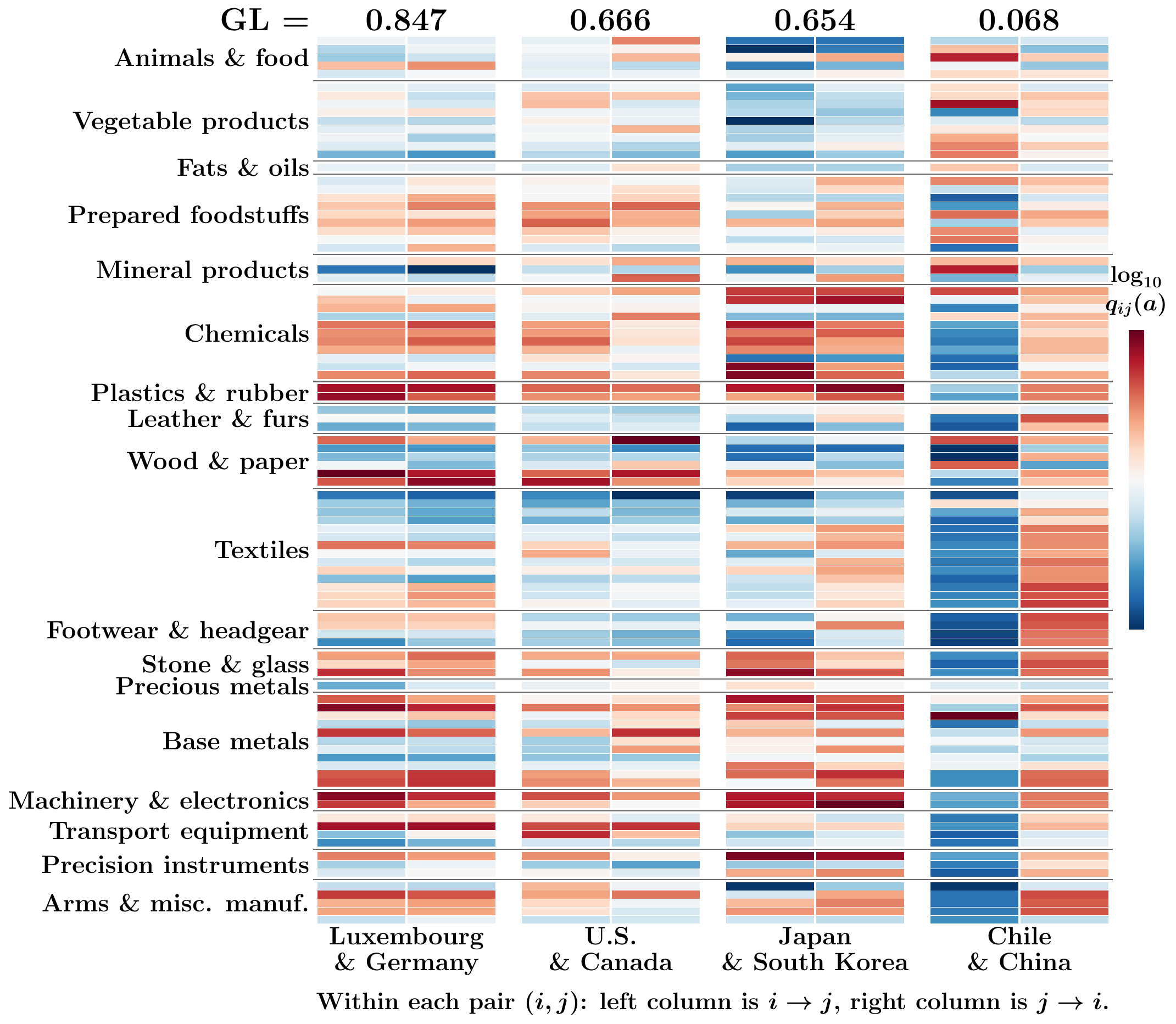}
        \caption{Four product-level examples.}
        \label{fig:4pairs}
    \end{subfigure}
\caption{\textbf{Grubel--Lloyd structure of trade pairs from BPRGTF reconstructions.} \textbf{Left:} The most symmetric trade pairs (GL near 1) are dominated by neighboring European economies, while the most asymmetric pairs (GL near 0) involve trade-specialized partners such as Chile--China and Brazil--Russia. \textbf{Right:} Four product-level examples span the spectrum from high (Luxembourg--Germany, GL $= 0.847$) to low (Chile--China, GL $= 0.068$) symmetry: Luxembourg--Germany exhibits nearly mirrored trade across all product categories, whereas Chile--China is dominated by a few categories flowing in one direction. Within each pair, the left column shows $i \to j$ flows and the right column shows $j \to i$ flows.}
    \label{fig:trade_overview}
\end{figure}

\textbf{Two-way trade balance.} For each ordered direction $i\to j$, the decomposition implies a product distribution $q_{ij}(a) = \sum_{k=1}^{K} \pi_{ij}(k) \tilde\theta^{(3)}_{ak}$ with $\pi_{ij}(k) \propto \tilde \theta^{(1)}_{ik} \tilde\theta^{(2)}_{jk}$, i.e.\ the share of $i$'s exports to $j$ in each product $a$ as attributed by the latent channels. We measure compositional symmetry by the Grubel--Lloyd index $\mathrm{GL}_{ij} = 1 - \tfrac{1}{2}\sum_a |q_{ij}(a) - q_{ji}(a)| \in [0, 1]$, with $\mathrm{GL}_{ij}\to 1$ for similar product baskets in both directions (intra-industry trade) and $\mathrm{GL}_{ij}\to 0$ for disjoint baskets (inter-industry specialization). The classical index \citep{grubel_lloyd1975} is computed on raw bilateral flows, while ours operates on the rank-$K$ latent reconstruction $q_{ij}$, suppressing flow noise and isolating structural composition. Among pairs whose smaller direction of trade lies in the top 5\% of directional volumes (\Cref{tab:gl}), the symmetric end is dominated by pairs within EU, especially Visegr\'ad countries, the canonical intra-industry pattern among similar economies 
\citep{helpman_krugman1985}, 
while the asymmetric end pairs commodity exporters (Chile, Brazil, Peru, Argentina) with industrial economies (China, U.S., Japan, Korea, Russia), reproducing the Heckscher–Ohlin pattern, in which factor-endowment differences drive trade between resource-abundant and capital-abundant economies \citep{feenstra2016advanced}.
\Cref{fig:4pairs} shows $q_{ij}$ and $q_{ji}$ at the HS2 product level for four pairs. For Luxembourg and Germany, both directions trade similar product mixes and GL is high, and for Chile and China, each direction concentrates in different products and GL is low. These two extremes correspond to the intra- and inter-industry regimes that \citet{kim2020measuring} obtain by clustering compositional profiles formed directly from observed dollar trade. More details are in Supplementary~\ref{sec:supp_gl}.

\textbf{Opposite time-factor trends.}
Most time factors exhibit gradual growth consistent with expanding global trade, with deviations around major shocks. A subset of components, however, displays opposing time trends despite sharing similar importer sets and product families. We interpret these through three lenses.

\emph{(i) Market share reallocation across latent export components.} \Cref{fig:us-machine} shows the U.S.-centered component for manufactures (HS84/85/90) into Japan, Korea, and Canada declining in relative prominence, while China- and Germany-centered components into overlapping markets rise. This reflects a loss of relative weight in the non-negative decomposition rather than an absolute decline in U.S. exports; the German time factor also rises more slowly than China's, consistent with U.S. manufacturing losing competitiveness to Chinese imports \citep{pierce_schott_2016} and China's rising export market share in medium- and high-tech goods, where Germany has lost ground since 2017 \citep{meinen_nagengast_2025}.

\begin{figure}[t!]
    \centering
    \begin{subfigure}[t]{0.50\textwidth}
        \centering
        \includegraphics[width=\linewidth]{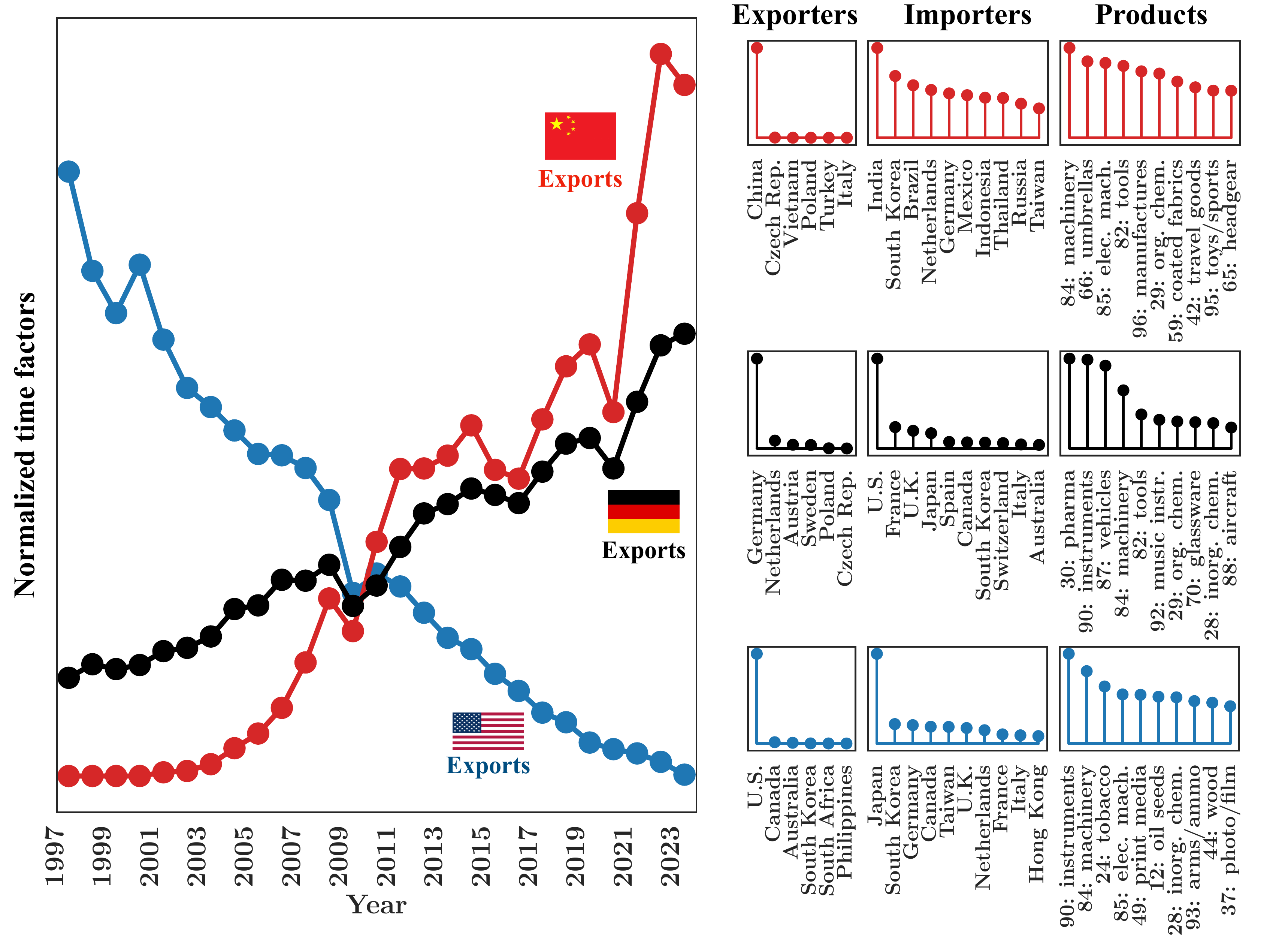}
        \caption{US-machine market competition}
        \label{fig:us-machine}
    \end{subfigure}\hfill
    \begin{subfigure}[t]{0.50\textwidth}
        \centering
        \includegraphics[width=\linewidth]{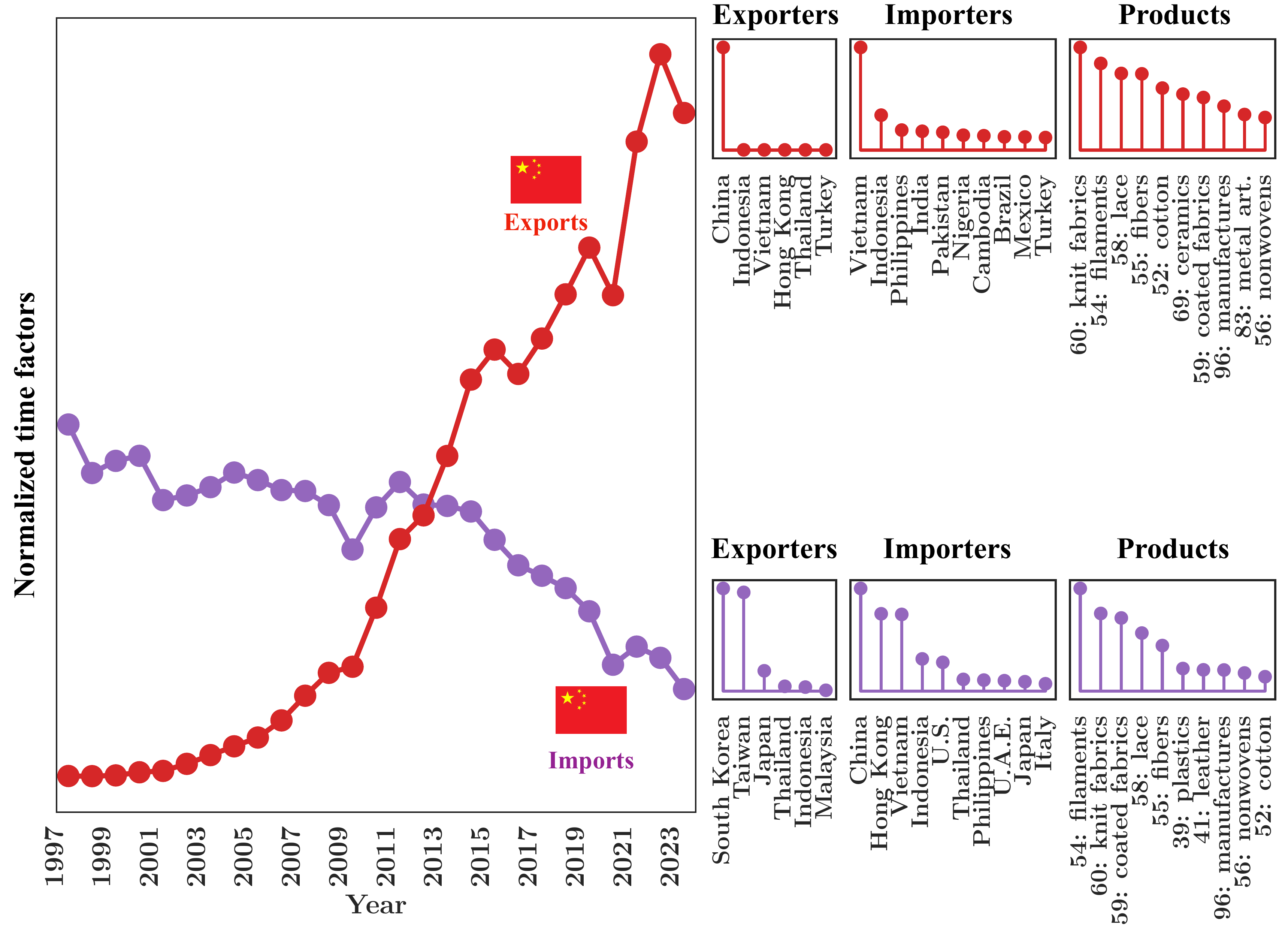}
        \caption{China-textiles industry}
        \label{fig:chn-textiles}
    \end{subfigure}
\caption{\textbf{Opposite time-factor trends in trade components.} Declining components are paired with rising components in overlapping markets and product families, revealing market-share reallocation and supply-chain restructuring. \textbf{Left:} A U.S.-centered machinery export component (blue) declines while China- and Germany-centered components (red and black) into overlapping destinations rise. \textbf{Right:} China's textile imports (purple) decline while its textile exports (red) rise over the same period, reflecting a transition from net importer to dominant exporter.}
    \label{fig:2decline}
\end{figure}

\emph{(ii) Trade-routing and policy changes.}
A China$\to$Hong Kong textiles component (Supplementary~\ref{app:china_hk}) drops sharply around 2005, consistent with reduced re-exporting, the routing of Chinese goods through Hong Kong to third markets \citep{hktid_mainland_entrepot}.

\emph{(iii) Supply-chain restructuring.}
\Cref{fig:chn-textiles} shows declining exports of textile intermediates (HS52/54/55/56/58/59/60) from Japan, Korea, and Thailand to China alongside rising Chinese exports of overlapping intermediates to Vietnam, Cambodia, and Indonesia. This is consistent with China's growing role as a regional supplier of fibers and fabrics to garment manufacturers across the textiles and clothing value chain \citep{wto_gvc_textiles_clothing_2024}.

\section{Conclusion}
\label{sec:conclusion}

We proposed Bayesian Poisson-randomized Gamma tensor factorization for semi-continuous multiway data which combines low-rank nonnegative CP structure with a compound Poisson--Gamma likelihood and mode-specific dispersion. The model extends Bayesian Poisson tensor factorization for count data \citep{schein2015bayesian} to semi-continuous observations, complementing existing extensions to binary data \citep{zhou2015infinite}, compound Poisson exponential-dispersion families \citep{BasbugEngelhardt2016HCPF}, and bounded responses on [0,1] \citep{albert2024doubly}.~\looseness=-1


A central computational contribution is the treatment of the binary activity indicators governing the zero-positive boundary. Their exact conditional involves a modified Bessel function of the first kind, which is costly to evaluate at scale. We introduce a partially collapsed sampler that avoids marginalizing over exponentially many activation configurations, and derive an asymptotic Bernoulli--Poisson approximation to the Bessel-family conditional that is accurate in the regime most common in practice. This reduces sampling to a single exponential evaluation per entry, enabling GPU-accelerated inference for tensors with tens of millions of entries. Applied to 60 million international trade flows, the model outperforms competing tensor baselines in out-of-sample prediction and recovers 125 interpretable components spanning exporters, importers, products, and time.

The proposed framework opens several future directions. The rank is treated as fixed, where a natural extension is to infer it from the data in a nonparametric Bayesian manner, as in \citet{zhao2015bayesian} and \citet{zhou2015infinite}. The temporal factors are modeled without temporal structure, treated as independent and exchangeable across years, and incorporating smoothness priors or regularization across adjacent time points would capture gradual trends between periods. Dyadic covariates such as bilateral distance, tariff schedules, and preferential trade agreements could enter through the latent Poisson rate tensor, connecting the model to structural gravity frameworks and tensor regression with structured coefficient priors \citep{guhaniyogi2017bayesian}. Beyond monetary-valued data, the framework applies broadly to any zero-inflated, positive-continuous multiway array, including drug-response tensors in pharmacology, species-abundance data in ecology, and healthcare utilization records in epidemiology.

{\renewcommand{\baselinestretch}{1}\selectfont
\bibliographystyle{plainnat}
\bibliography{references}
}

\clearpage
\appendix
\renewcommand{\thesection}{S\arabic{section}}
\renewcommand{\thesubsection}{S\arabic{section}.\arabic{subsection}}
\renewcommand{\thefigure}{S\arabic{figure}}
\renewcommand{\thetable}{S\arabic{table}}
\renewcommand{\theequation}{S\arabic{equation}}
\renewcommand{\thetheorem}{S\arabic{theorem}}
\setcounter{section}{0}
\setcounter{figure}{0}
\setcounter{table}{0}
\setcounter{equation}{0}
\setcounter{theorem}{0}


\section{Proof of identifiability}
\label{app:PRGidentifiability}

\begin{proof}[Proof of Lemma~\ref{lem:ident_entrywise_mu_beta} in the main text]
Suppose $F_{\lambda,\beta}=F_{\lambda',\beta'}$.

First, by the definition of PRG$(\lambda,\beta)$,
\begin{equation} \label{eq:prg0}
\eta\mid \lambda \sim \mathrm{Pois}(\lambda), 
\qquad
Y\mid \eta,\beta \sim 
\begin{cases}
\delta_0, & \eta=0,\\
\mathrm{Gamma}(\eta,\beta), & \eta>0,
\end{cases}
\end{equation}
we have 
\[
\mathbb{P}(Y=0\mid \lambda,\beta)=\mathbb{P}(\eta=0\mid \lambda)=e^{-\lambda},
\]
and similarly $\mathbb{P}(Y=0\mid \lambda',\beta')=e^{-\lambda'}$. Equality in distribution implies
$e^{-\lambda}=e^{-\lambda'}$, hence $\lambda=\lambda'$.

Next, for any $t>0$, the Laplace transform of $Y$ under \cref{eq:prg0} satisfies
\begin{align*}
\mathbb{E}\!\left[e^{-tY}\mid \eta,\beta\right]
&=
\begin{cases}
1, & \eta=0,\\[1mm]
\left(\dfrac{\beta}{\beta+t}\right)^{\eta}, & \eta>0,
\end{cases}
\end{align*}
so by iterated expectation and the probability generating function of a Poisson random variable,
\begin{align*}
\mathbb{E}\!\left[e^{-tY}\mid \lambda,\beta\right]
&=\mathbb{E}\!\left[\left(\frac{\beta}{\beta+t}\right)^{\eta}\right]
=\exp\!\left(\lambda\left(\frac{\beta}{\beta+t}-1\right)\right)
=\exp\!\left(-\lambda\,\frac{t}{\beta+t}\right).
\end{align*}
Since $F_{\lambda,\beta}=F_{\lambda',\beta'}$, the Laplace transforms coincide for all $t>0$; using $\lambda=\lambda'$ from above,
\[
\exp\!\left(-\lambda\,\frac{t}{\beta+t}\right)=\exp\!\left(-\lambda\,\frac{t}{\beta'+t}\right)
\quad \forall\, t>0,
\]
which implies $\frac{t}{\beta+t}=\frac{t}{\beta'+t}$ for all $t>0$, and hence $\beta=\beta'$.
\end{proof}

\section{The zero-truncated Poisson--Gamma mixture density}
\label{appendix:RGplus}

We adopt the notation from the main text. We define $y \sim \mathrm{PRG}_+(\lambda^{(k)}_{\mathbf{i}}, \beta_{\mathbf{i}})$ if $y \mid \eta^{(k)} \sim \mathrm{Gam}(\eta^{(k)}, \beta_{\mathbf{i}})$ where $\eta^{(k)} \sim \mathrm{Pois}_+(\lambda^{(k)}_{\mathbf{i}})$. The density of $y$ is obtained by marginalizing over $\eta^{(k)}$, with $m$ used as a local summation index:
\begin{align*}
\mathrm{PRG}_+\!\left(y;\lambda^{(k)}_{\mathbf{i}},\beta_{\mathbf{i}}\right)
&= \sum_{m=1}^{\infty} \mathrm{Gam}\!\left(y;m,\beta_{\mathbf{i}}\right)\cdot \mathrm{Pois}_+\!\left(m;\lambda^{(k)}_{\mathbf{i}}\right)\\
&= \sum_{m=1}^{\infty} \frac{\beta_{\mathbf{i}}^{m}}{\Gamma(m)}y^{m-1} e^{-\beta_{\mathbf{i}} y} \cdot \frac{\left(\lambda^{(k)}_{\mathbf{i}}\right)^{m} e^{-\lambda^{(k)}_{\mathbf{i}}}}{m!\left(1-e^{-\lambda^{(k)}_{\mathbf{i}}}\right)}\\
&= \frac{e^{-\beta_{\mathbf{i}} y}}{y\!\left(e^{\lambda^{(k)}_{\mathbf{i}}}-1\right)} \sum_{m=1}^{\infty} \frac{\left(\beta_{\mathbf{i}}\,\lambda^{(k)}_{\mathbf{i}}\,y\right)^{m}}{m!\,\Gamma(m)}\\
&= \frac{e^{-\beta_{\mathbf{i}} y}\sqrt{\beta_{\mathbf{i}}\,\lambda^{(k)}_{\mathbf{i}}\,y}}{y\!\left(e^{\lambda^{(k)}_{\mathbf{i}}}-1\right)}\,I_{-1}\!\left(2\sqrt{\beta_{\mathbf{i}}\,\lambda^{(k)}_{\mathbf{i}}\,y}\right)
\end{align*}

\section{Evidence lower bound (ELBO)}
\label{appendix:elbo}

For observations $\mathcal{Y}$ and latent variables $\mathcal{Z}= \{y_{\mathbf{i}}^{(k)}\}_{\mathbf{i}, k} \;\cup\; \{\eta_{\mathbf{i}}^{(k)}\}_{\mathbf{i}, k} \;\cup\; \{\Theta^{(m)}\}_{m}$, the evidence lower bound (ELBO) rewrites the log marginal likelihood as
\begin{equation}
\label{eq:elbo-basic}
\mathrm{ELBO}(q)=\E{q}{\log p(\mathcal{Y},\mathcal{Z})}-\E{q}{\log q(\mathcal{Z})}.
\end{equation}

In our hybrid algorithm, the binary activity indicators $\mathcal{A}= \{ b^{(k)}_{\mathbf{i}} \}_{\mathbf{i},\,k}$ are sampled each sweep, and thus the variational approximation is applied to the remaining latent variables conditional on $\mathcal{A}$.

\textbf{Variational family.}
Although $\{y_{\mathbf{i}}^{(k)}\}_{\mathbf{i},k}$ enters $\mathcal{Z}$ directly, we work with the equivalent reparametrization $y_{\mathbf{i}}^{(k)} = \pi_{\mathbf{i}}^{(k)}\,y_{\mathbf{i}}$ from the main text, treating the total $y_{\mathbf{i}} \sim \mathrm{Gam}(\cdot)$ and proportions $(\pi_{\mathbf{i}}^{(k)})_{k\in\Delta_{\mathbf{i}}} \sim \mathrm{Dir}(\cdot)$ as the latent quantities.
We adopt the mean-field factorization
\[
q\!\left(\{\pi_{\mathbf{i}}\}_{\mathbf{i}},\, \{\eta^{(k)}_{\mathbf{i}}\}_{\mathbf{i},k},\, \{\Theta^{(m)}\}_{m}\right)
=
\prod_{\mathbf{i}:\,y_{\mathbf{i}}>0}
\prod_{k\in \Delta_{\mathbf{i}}}
q(\pi_\mathbf{i}^{(k)})
\ \cdot\
\prod_{\mathbf{i}:\,y_{\mathbf{i}}>0}
\prod_{k\in \Delta_{\mathbf{i}}}
q\!\left(\eta^{(k)}_\mathbf{i}\right)
\ \cdot\
\prod_{m=1}^M\prod_{k=1}^K\prod_{i_m=1}^{I_m}
q\!\left(\theta^{(m)}_{i_m k}\right),
\]
where $q(\pi_{\mathbf{i}})$ is Dirichlet supported on the active set $\Delta_{\mathbf{i}}$,
$q(\eta^{(k)}_{\mathbf{i}})$ is a truncated Bessel variational factor for $k\in\Delta_{\mathbf{i}}$ and degenerates to a point mass at $0$ for $k\notin\Delta_{\mathbf{i}}$,
and each $q(\theta^{(m)}_{i_m k})$ is Gamma.

\textbf{Joint distribution under augmentation.}
For each nonzero entry $y_{\mathbf{i}}>0$, let
$s_{\mathbf{i}}=\sum_{k\in\Delta_{\mathbf{i}}}\eta^{(k)}_{\mathbf{i}}$.
Using the Gamma-additivity or Dirichlet allocation identity, the likelihood contribution can be written as
\[
\mathbb{P} \big(y_{\mathbf{i}},\pi_{\mathbf{i}}\mid \eta, b\big)
=
\text{Gam}\!\big(y_{\mathbf{i}};\, s_{\mathbf{i}},\,\beta_{\mathbf{i}}\big)\,
\text{Dir}\!\big(\pi_{\mathbf{i}};\,(\eta^{(k)}_{\mathbf{i}})_{k\in\Delta_{\mathbf{i}}}\big),
\]
up to constants that do not depend on variational parameters.
The remaining components of the joint are
\[
\eta^{(k)}_{\mathbf{i}}\mid \Theta \sim \text{Pois}\!\big(\lambda^{(k)}_{\mathbf{i}}\big),
\qquad
\lambda^{(k)}_{\mathbf{i}}=\prod_{m=1}^M \theta^{(m)}_{i_m k},
\qquad
\theta^{(m)}_{i_m k}\sim \text{Gam}(\alpha_0,\alpha_0\nu^{(m)}).
\]

Let $\Omega_+=\{\mathbf{i}: y_{\mathbf{i}}>0\}$. Substituting the model and the variational family into \cref{eq:elbo-basic} yields

\begin{align}
\label{eq:elbo-decomp}
\mathrm{ELBO}
&=
\underbrace{\sum_{\mathbf{i}\in\Omega_+}\E{q}{\log \mathrm{Gam}\big(y_{\mathbf{i}};\, s_{\mathbf{i}},\,\beta_{\mathbf{i}}\big)}}_{\text{(likelihood: total magnitude)}} \\
&\quad+\underbrace{\sum_{\mathbf{i}\in\Omega_+}\E{q}{\log \mathrm{Dir} \big(\pi_{\mathbf{i}};\,\{\eta^{(k)}_{\mathbf{i}}\}_{k\in\Delta_{\mathbf{i}}}\big)}}_{\text{(allocation: proportions)}} \nonumber\\
&\quad+\underbrace{\sum_{k=1}^K\sum_{\mathbf{i}\in\Omega_+} \E{q}{\log \mathrm{Pois} \big(\eta^{(k)}_{\mathbf{i}};\,\lambda^{(k)}_{\mathbf{i}}\big)}}_{\text{(Poisson layer)}} \nonumber\\
&\quad+ \underbrace{\sum_{m=1}^M \sum_{k=1}^K\sum_{i_m=1}^{I_m} \E{q}{\log \mathrm{Gam} \big(\theta^{(m)}_{i_m k};\,\alpha_0,\alpha_0\nu^{(m)}\big)}}_{\text{(priors)}} \nonumber\\
&\quad- \underbrace{\sum_{\mathbf{i}\in\Omega_+}\E{q}{\log q(\pi_{\mathbf{i}})} -\sum_{k=1}^K\sum_{\mathbf{i}\in\Omega_+}\E{q}{\log q(\eta^{(k)}_{\mathbf{i}})} -\sum_{m=1}^M \sum_{k=1}^K\sum_{i_m=1}^{I_m}\E{q}{\log q(\theta^{(m)}_{i_m k})}}_{\text{(entropies)}}.\nonumber
\end{align}

For implementation, the expectations in \cref{eq:elbo-decomp} are evaluated using the moments and log-moments of the variational families (Dirichlet, truncated Bessel, and Gamma). In particular, $\E{q}{\log \theta^{(m)}_{i_m k}}$ and $\E{q}{\theta^{(m)}_{i_m k}}$ have closed forms under the Gamma variational factors, and the Bessel terms are computed using the corresponding special functions.

\section{Proof of Proposition 1}
\label{appendix:proof_proposition1}

For completeness we first restate Proposition~\ref{prop1} in the main text with the explicit forms of the two conditional Bernoulli probabilities.

\begin{proposition}[Bernoulli--Poisson approximation; full statement]
Let $\lambda,s,\beta>0$, and suppose
\[
Y_1\sim\mathrm{Gam}(s,\beta),\qquad
Y_2\mid \eta \sim
\begin{cases}
\delta_0, & \eta=0,\\
\mathrm{Gam}(\eta,\beta), & \eta>0,
\end{cases}
\qquad
\eta\sim\mathrm{Pois}(\lambda),
\qquad
Y=Y_1+Y_2.
\]
For a fixed observation $Y=y$, let $Z=\mb{1}(Y_2>0)$. Then, conditional on
$(\lambda,s,\beta,Y=y)$,
\[
Z\sim \mathrm{Bernoulli}\bigl(p(\lambda,s,\beta,y)\bigr),\qquad
p(\lambda,s,\beta,y)
=1-\frac{(\sqrt{\beta y\lambda})^{\,s-1}}
{\Gamma(s)\,I_{s-1}\!\bigl(2\sqrt{\beta y\lambda}\bigr)}.
\]
If $X\mid(\lambda,s,\beta,Y=y)\sim\mathrm{Pois}\bigl(\lambda\beta y/s\bigr)$ and
$\widetilde Z=\mb{1}(X>0)$, then
\[
\widetilde Z\mid (\lambda,s,\beta,Y=y)\sim
\mathrm{Bernoulli}\bigl(\widetilde p(\lambda,s,\beta,y)\bigr),\qquad
\widetilde p(\lambda,s,\beta,y)=1-\exp\!\bigl(-\lambda\beta y/s\bigr).
\]
As $s\to\infty$, $\bigl|p(\lambda,s,\beta,y)-\widetilde p(\lambda,s,\beta,y)\bigr|\to 0$,
and hence, since $d_{\mathrm{TV}}(\mathrm{Bernoulli}(p),\mathrm{Bernoulli}(q))=|p-q|$,
\[
d_{\mathrm{TV}}\!\left(\mathrm{Bernoulli}\bigl(p(\lambda,s,\beta,y)\bigr),\,
\mathrm{Bernoulli}\bigl(\widetilde p(\lambda,s,\beta,y)\bigr)\right)\longrightarrow 0.
\]
\end{proposition}

\begin{proof}
Because $d_{\mathrm{TV}}(\mathrm{Bernoulli}(p),\mathrm{Bernoulli}(q))=|p-q|$, it suffices to prove that $\bigl|p(\lambda,s,\beta,y)-\widetilde p(\lambda,s,\beta,y)\bigr|\to 0$ as $s\to\infty$. Write $x=\lambda\beta y>0$ and $m=s$, and define
\[A_m=\frac{(\sqrt{x})^{\,m-1}}{\Gamma(m)\,I_{m-1}(2\sqrt{x})},\qquad B_m=\exp(-x/m),
\]
so that $p=1-A_m$, $\widetilde p=1-B_m$, and the claim reduces to showing $|A_m-B_m|\to 0$ as $m\to\infty$.

Using the series representation of the modified Bessel function,
\[
I_{m-1}(2\sqrt{x})=\sum_{k=0}^{\infty}\frac{(\sqrt{x})^{\,2k+m-1}}{k!\,\Gamma(m+k)},
\]
we obtain
\[
A_m^{-1}=\sum_{k=0}^{\infty}\frac{x^k}{k!}\,\frac{\Gamma(m)}{\Gamma(m+k)},
\qquad
B_m^{-1}=\exp(x/m)=\sum_{k=0}^{\infty}\frac{x^k}{k!\,m^k}.
\]
By the Pochhammer identity $\Gamma(m+k)/\Gamma(m)=\prod_{j=0}^{k-1}(m+j)$, for every $m\geq 1$ and $k\geq 0$,
\begin{equation}
\label{eq:pochhammer-bound}
0\;\leq\;\frac{\Gamma(m)}{\Gamma(m+k)}
=\prod_{j=0}^{k-1}\frac{1}{m+j}
\;\leq\;\frac{1}{m^k},
\end{equation}
so every term of the series
\[
B_m^{-1}-A_m^{-1}
=\sum_{k=0}^{\infty}\frac{x^k}{k!}
\left[\frac{1}{m^k}-\frac{\Gamma(m)}{\Gamma(m+k)}\right]
\]
is nonnegative and dominated by $x^k/k!$, which is summable (with sum $e^x$, independent of $m$). For each fixed $k\geq 0$,
\[
\frac{\Gamma(m)}{\Gamma(m+k)}\cdot m^k
=\prod_{j=0}^{k-1}\frac{m}{m+j}\longrightarrow 1
\qquad\text{as }m\to\infty,
\]
so each term of the series tends to $0$. By Tannery's theorem (dominated convergence for series),
\[
B_m^{-1}-A_m^{-1}\longrightarrow 0
\qquad\text{as }m\to\infty.
\]

From \cref{eq:pochhammer-bound}, $A_m^{-1}\leq B_m^{-1}$ and hence $B_m\leq A_m\leq 1$; combined with $B_m=\exp(-x/m)\to 1$, this forces $A_m\to 1$ and $A_m B_m\to 1$. Therefore
\[
|A_m-B_m|
=A_m B_m\,\bigl|A_m^{-1}-B_m^{-1}\bigr|
\longrightarrow 0,
\]
which completes the proof.
\end{proof}

\section{Updates of Gamma rates}
\label{app:beta_earlystop}

In practice, we update the Gamma rate parameters $\{\beta_{i_a}\}$ by empirical Bayes using Eq.~\eqref{eq:beta_a} in the main text only during the first 50 adaptation iterations prior to sampling the indicators $b$ and then hold them fixed thereafter. In simulations, fully iterating the EB updates improves recovery of the absolute $\beta_{i_a}$ values, but can slightly degrade recovery of the low-rank structure; early stopping yields more stable factor estimates and better overall predictive performance.

\Cref{fig:beta_earlystop} illustrates this tradeoff for $A=30$ independent rates: early stopping induces a level shift in $\hat\beta_{i_a}$ but closely preserves the relative variation across indices (ordering and scale differences). \Cref{tab:synthetic_K_20_with_earlystop} compares early-stopped and fully-updated EB within the same simulation design in Table~\ref{tab:synthetic_K_20} in the main text, while the fully-updated variant can improve some nonzero-magnitude metrics, it typically worsens structural recovery (e.g., higher Hamming loss). We therefore use early-stopped EB updates throughout the main text.

\begin{figure}[t]
    \centering
    \includegraphics[width=0.7\linewidth]{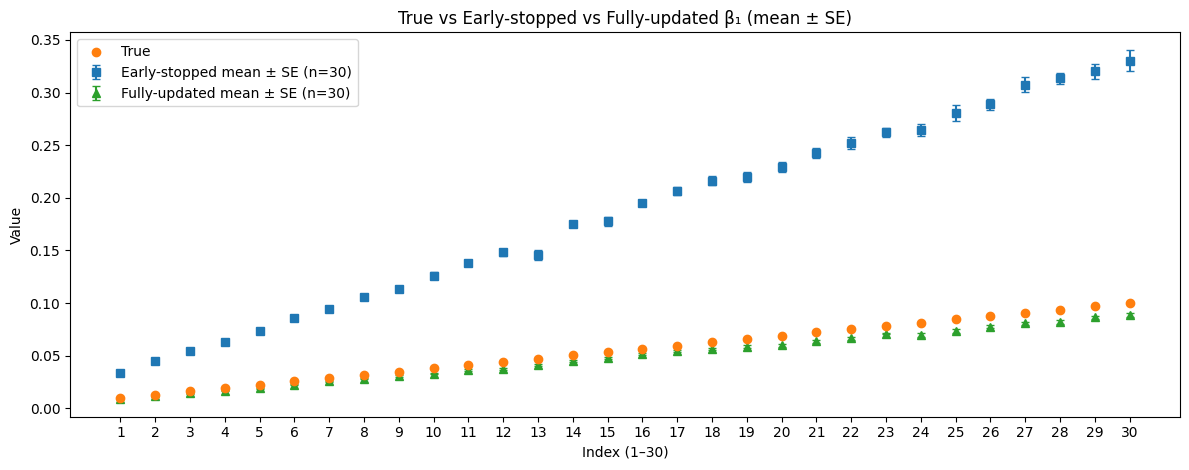}
    \caption{\textbf{Early-stopped updates recover relative heterogeneity but exhibit level bias.}
    Synthetic experiment with $A=30$ independent $\beta_a$ parameters. Orange: true $\beta_a$ values.
    Blue: posterior means (with $\pm$ SE across replicates) under early-stopped EB updates.
    Early stopping induces a systematic level shift but preserves the relative variation across $a$ (ordering and scale differences), which is often the primary role of $\beta_a$ in capturing slice-specific dispersion.}
    \label{fig:beta_earlystop}
\end{figure}

\begin{table}[!t]
\caption{Effect of early-stopped versus fully-updated empirical-Bayes updates of $\beta_{i_a}$ in simulation. The fully-updated variant can improve absolute recovery of $\beta_{i_a}$ and some nonzero-magnitude errors, but often increases structural error (e.g., Hamming loss), motivating our early-stopping schedule. Row-wise optima are bolded.}
\label{tab:synthetic_K_20_with_earlystop}
\centering
{\renewcommand{\baselinestretch}{1}\selectfont
\setlength{\tabcolsep}{5pt}
\begin{tabular}{@{}lccccc@{}}
\toprule
 & \multicolumn{2}{c}{Baselines} & \multicolumn{3}{c}{Proposed (BPRGTF)} \\
\cmidrule(lr){2-3} \cmidrule(lr){4-6}
Metric & NTF-LS & NTF-$\beta$
  & \makecell{shared\\($\beta_{\mathbf{i}}{=}\beta$)}
  & \makecell{local\\($\beta_{\mathbf{i}}{=}\beta_a$)\\(early)}
  & \makecell{local\\($\beta_{\mathbf{i}}{=}\beta_a$)\\(fully)} \\
\midrule
\multicolumn{6}{c}{\textit{Working rank} $K=15$} \\
\midrule
MAE ($\downarrow$)       & 20.6 (0.5)    & 50.3 (1.6)    & 22.3 (0.7)    & \textbf{18.2 (0.5)}    & 21.7 (0.6) \\
MAE-NZ ($\downarrow$)    & 56.3 (0.9)    & 97.7 (2.6)    & 57.5 (1.0)    & \textbf{55.3 (0.8)}    & 56.2 (1.0) \\
MArel-NZ ($\downarrow$)  & 4.6 (0.3)     & 7.9 (0.4)     & 5.7 (0.3)     & \textbf{3.2 (0.1)}     & 5.0 (0.2)  \\
HAM-Z ($\downarrow$)     & 0.699 (0.010) & 0.986 (0.002) & 0.945 (0.007) & \textbf{0.432 (0.008)} & 0.868 (0.009) \\
AUC ($\uparrow$)         & 0.859 (0.005) & 0.682 (0.004) & 0.894 (0.002) & \textbf{0.906 (0.002)} & 0.882 (0.003) \\
\midrule
\multicolumn{6}{c}{\textit{Working rank} $K=20$} \\
\midrule
MAE ($\downarrow$)       & 19.8 (0.4)    & 48.8 (1.6)    & 20.1 (0.5)    & \textbf{16.8 (0.4)}    & 19.1 (0.4) \\
MAE-NZ ($\downarrow$)    & 54.3 (0.7)    & 95.8 (2.6)    & 53.7 (0.6)    & 51.6 (0.6)             & \textbf{51.5 (0.6)} \\
MArel-NZ ($\downarrow$)  & 4.7 (0.3)     & 8.5 (0.6)     & 5.4 (0.3)     & \textbf{3.2 (0.2)}     & 4.7 (0.2)  \\
HAM-Z ($\downarrow$)     & 0.674 (0.011) & 0.982 (0.002) & 0.889 (0.009) & \textbf{0.339 (0.007)} & 0.754 (0.012) \\
AUC ($\uparrow$)         & 0.882 (0.004) & 0.693 (0.004) & 0.918 (0.002) & \textbf{0.938 (0.002)} & 0.914 (0.002) \\
\midrule
\multicolumn{6}{c}{\textit{Working rank} $K=25$} \\
\midrule
MAE ($\downarrow$)       & 19.4 (0.4)    & 46.8 (1.6)    & 19.3 (0.5)    & \textbf{16.6 (0.4)}    & 18.3 (0.4) \\
MAE-NZ ($\downarrow$)    & 53.5 (0.7)    & 92.5 (2.4)    & 52.4 (0.6)    & 51.2 (0.5)             & \textbf{50.3 (0.6)} \\
MArel-NZ ($\downarrow$)  & 4.8 (0.3)     & 7.7 (0.4)     & 5.3 (0.3)     & \textbf{3.3 (0.2)}     & 4.7 (0.2)  \\
HAM-Z ($\downarrow$)     & 0.650 (0.010) & 0.978 (0.002) & 0.843 (0.011) & \textbf{0.322 (0.006)} & 0.699 (0.011) \\
AUC ($\uparrow$)         & 0.897 (0.004) & 0.704 (0.003) & 0.928 (0.002) & \textbf{0.943 (0.001)} & 0.924 (0.002) \\
\bottomrule
\end{tabular}
}
\end{table}

\section{Supplementary Material for Section 6}

\subsection{Preprocessing details}
\label{app:preprocessing}

We flag a reporter-year lapse whenever an exporter reports no positive exports in a given year across all partners and products, i.e., $\sum_{j,a}\mb{1}(y_{ijat}>0)=0$. In such cases, we treat the entire exporter-year block as missing rather than as genuine zero trade. For each $(i,j,a)$ series, we impute short internal gaps of length at most two years by linear interpolation in levels between the nearest observed years. If a single missing value occurs at the beginning or end of the sample window, we fill it using the nearest observed value. Longer gaps are left unimputed. We retain a binary observation mask that equals one for observed or short-gap-imputed cells and zero otherwise, so that structurally missing entries do not contribute to estimation. All remaining zeros are treated as observed zeros.

\Cref{fig:hist_supp} shows the distribution of zeros across exporter--importer--product triples. The distribution is sharply concentrated at 27 (the full sample length), indicating that most dyad--product series are zero throughout and that observed zeros likely combine true absence of trade with missingness.

\begin{figure}[!htbp]
    \centering
    \includegraphics[width=0.85\linewidth]{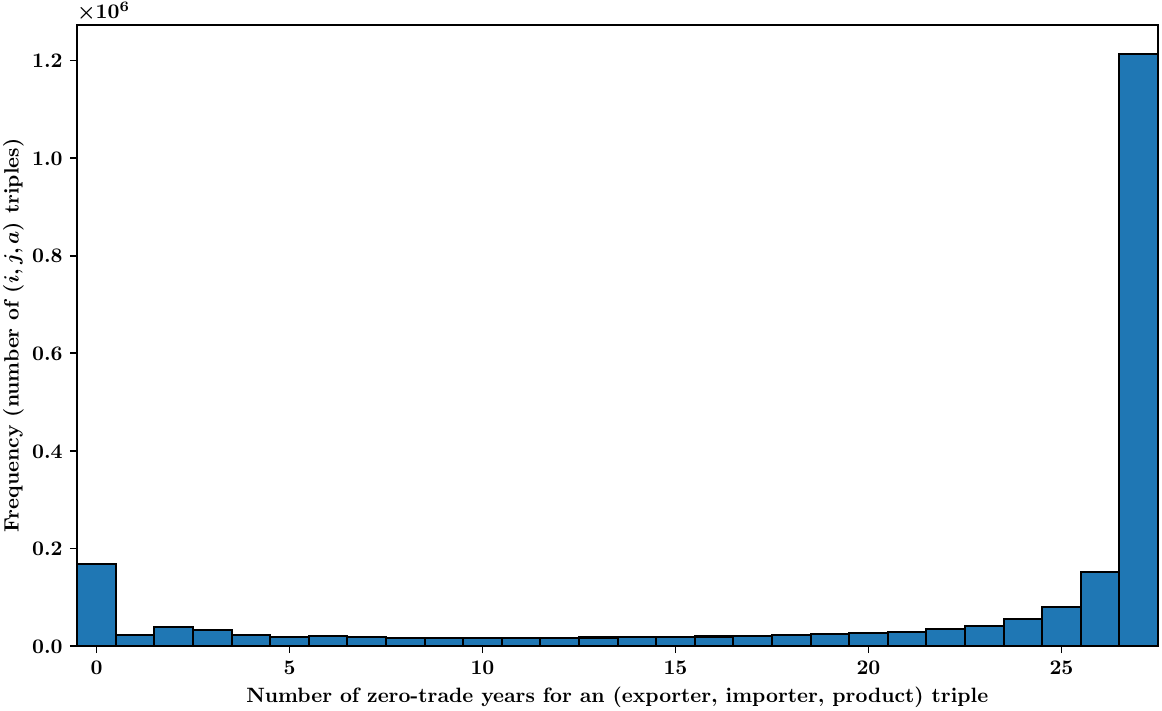}
    \caption{Histogram of $\sum_{t=1997}^{2023} \mathbf{1}(y_{ijat}=0)$ across all exporter--importer--product triples.}
    \label{fig:hist_supp}
\end{figure}

\subsection{Masking construction for out-of-sample evaluation}
\label{app:masking}

Within each held-out year slice $\mathcal{Y}_{t_0}$, the dense-block and complement masking schemes are defined as follows. Let $N_0 < N$ and define a dense core set $S_{t_0} \subset \{1,\dots,N\}$ of $N_0$ countries as those with the largest average trade volume in year $t_0$, where for each product $a$ trade volume is computed as the sum of exports and imports across all partners and then averaged over $a$. The dense block is the sub-tensor $y_{ijat_0}$ with $i,j \in S_{t_0}$ and $a = 1,\dots,A$. In the dense-block scenario, we observe this $N_0 \times N_0 \times A$ block and predict its complement; in the sparse-block (complement) scenario, we reverse the roles.

\Cref{fig:masking_supp} illustrates this construction on a single product--year slice ($a=9$, coffee and spices; $t_0=2023$). Countries are ranked by total trade intensity computed from the symmetrized matrix $y_{::at_0}+y_{::at_0}^\top$, after which the dense core corresponds to the top-left block.

\begin{figure}[!htbp]
  \centering
  \begin{subfigure}[t]{0.49\textwidth}
    \centering
    \includegraphics[width=\linewidth]{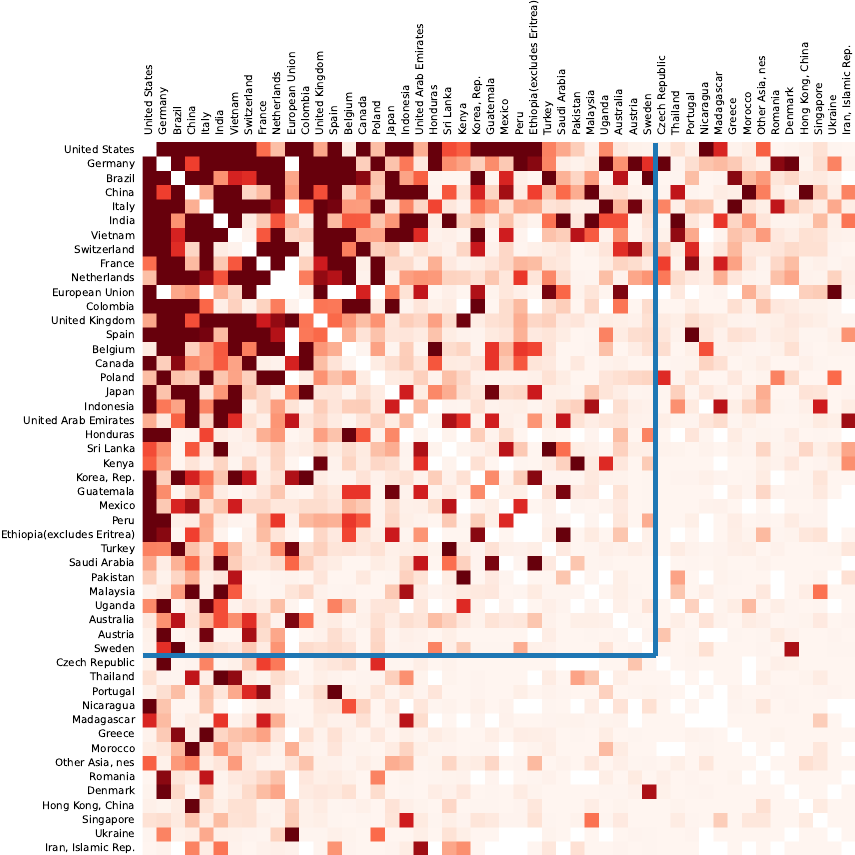}
    \caption{Symmetrized and ranked slice with the dense-block boundary marked (top $25\%\times 25\%$).}
  \end{subfigure}\hfill
  \begin{subfigure}[t]{0.49\textwidth}
    \centering
    \includegraphics[width=\linewidth]{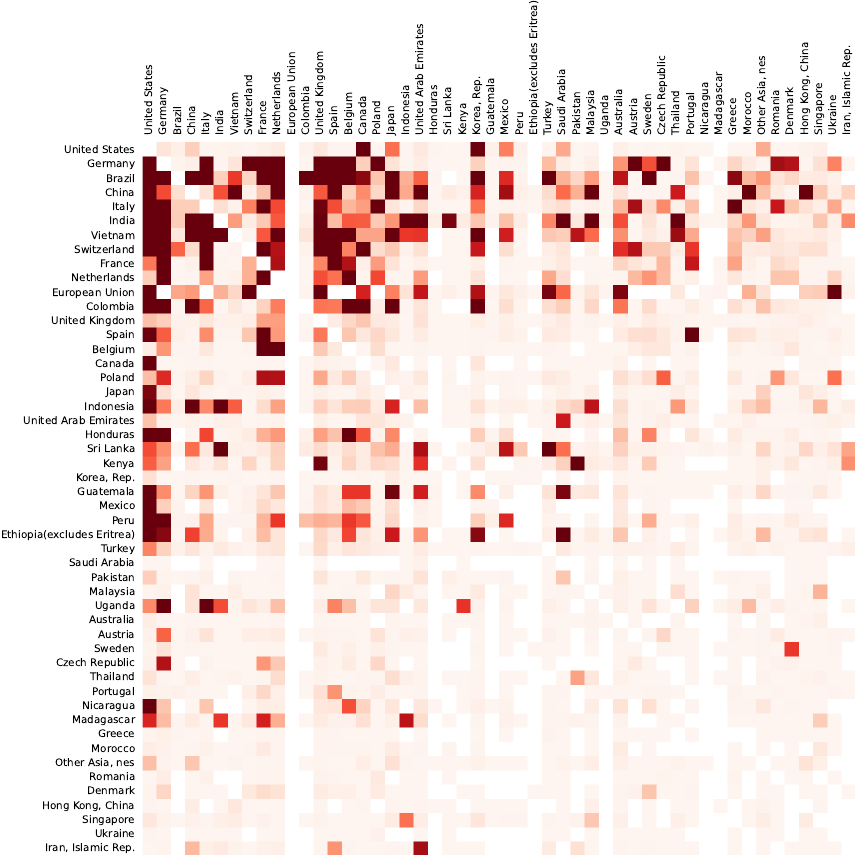}
    \caption{Raw directed slice under the same country ordering.}
  \end{subfigure}
  \caption{Illustration of the masking construction for one product--year slice ($a=9$, coffee and spices; $t_0=2023$). Each panel shows the $50\times 50$ submatrix of the country--country trade matrix after ranking by total trade intensity.}
  \label{fig:masking_supp}
\end{figure}

\subsection{Out-of-sample results at K=125}
\label{app:oos_table}

\Cref{tab:oos_trading_K125_supp} reports the full numerical results at rank $K=125$ for all three masking schemes.

\begin{table}[!t]
\caption{Out-of-sample predictive performance on the trade data at $K=125$. For each scenario, we report the variance-to-mean ratio (VMR) and nonzero density of the held-out set. BPRGTF with product-specific dispersion delivers the best overall performance.}
\label{tab:oos_trading_K125_supp}
\centering
{\renewcommand{\baselinestretch}{1}\selectfont
\footnotesize
\setlength{\tabcolsep}{4pt}
\begin{tabular}{@{}lcccc@{}}
\toprule
 & & & \multicolumn{2}{c}{BPRGTF} \\
\cmidrule(lr){4-5}
Metric & NTF-LS & NTF-$\beta$
  & \makecell{shared\\($\beta_{\mathbf{i}}{=}\beta$)}
  & \makecell{local\\($\beta_{\mathbf{i}}{=}\beta_a$)} \\
\midrule
\multicolumn{5}{l}{\textit{Complement held-out: size $=51{,}878{,}304$;\; density $=0.187$;\; VMR $=570{,}932.4$}} \\
\addlinespace[2pt]
MAE ($\downarrow$)       & 994.6           & 4834.6           & 2171.8           & \textbf{746.5}    \\
MAE-NZ ($\downarrow$)    & 4597.6          & 13610.1          & 6904.5           & \textbf{3717.4}   \\
MArel-NZ ($\downarrow$)  & 115.4           & 1137.4           & 588.8            & \textbf{40.1}     \\
HAM-Z ($\downarrow$)     & \textbf{0.554}  & 0.982            & 1.000            & 0.581             \\
AUC ($\uparrow$)         & 0.863           & 0.777            & 0.895            & \textbf{0.898}    \\
\midrule
\multicolumn{5}{l}{\textit{Dense held-out: size $=3{,}742{,}848$;\; density $=0.774$;\; VMR $=12{,}075{,}108.0$}} \\
\addlinespace[2pt]
MAE ($\downarrow$)       & 41007.6         & 66762.7          & 48042.9          & \textbf{37534.3}  \\
MAE-NZ ($\downarrow$)    & 51856.4         & 79840.5          & 60692.6          & \textbf{48159.1}  \\
MArel-NZ ($\downarrow$)  & 113.8           & 155.5            & 198.4            & \textbf{23.7}     \\
HAM-Z ($\downarrow$)     & 0.944           & 0.997            & 1.000            & \textbf{0.929}    \\
AUC ($\uparrow$)         & 0.811           & 0.790            & 0.851            & \textbf{0.872}    \\
\midrule
\multicolumn{5}{l}{\textit{Random held-out (averaged over 5 splits): size $=13{,}813{,}501$;\; density $=0.228$;\; VMR $=9{,}548{,}016.3$}} \\
\addlinespace[2pt]
MAE ($\downarrow$)       & 3076.3 (25.9)    & 5241.8 (154.7)   & 4746.4 (67.9)    & \textbf{2901.3 (42.2)}   \\
MAE-NZ ($\downarrow$)    & 12860.1 (101.5)  & 22226.2 (667.2)  & 17730.2 (285.5)  & \textbf{12536.1 (189.1)} \\
MArel-NZ ($\downarrow$)  & 117.9 (2.7)      & 127.6 (2.7)      & 423.0 (5.1)      & \textbf{32.8 (0.4)}      \\
HAM-Z ($\downarrow$)     & 0.558 (0.001)    & 0.601 (0.002)    & 1.000 (0.000)    & \textbf{0.547 (0.003)}   \\
AUC ($\uparrow$)         & 0.884 (0.000)    & 0.915 (0.001)    & 0.912 (0.000)    & \textbf{0.917 (0.000)}   \\
\bottomrule
\end{tabular}
}
\end{table}

\subsection{Product-specific Gamma rates and empirical association with trade-flow scale}
\label{subsec:beta_values}

For each product category $a$, we summarize the magnitude and dispersion of observed trade flows using only the non-zero entries of $\mathcal{Y}_{::a:}$. Let $\mathcal{S}_a \;=\; \{(i,j,t): y_{ijat}>0\}$ and 
\[
\bar{Y}^{(\mathrm{nz})}_a \;=\; \frac{1}{|\mathcal{S}_a|}\sum_{(i,j,t)\in\mathcal{S}_a} y_{ijat},
\qquad
\mathrm{Var}\!\left(Y^{(\mathrm{nz})}_a\right) \;=\; \frac{1}{|\mathcal{S}_a|-1}\sum_{(i,j,t)\in\mathcal{S}_a}\!\Big(y_{ijat}-\bar{Y}^{(\mathrm{nz})}_a\Big)^2 .
\]
We then assess the association between $\beta_a$ and these summaries using Spearman's rank correlation, which is sensitive to monotone relationships. Across products, $\beta_a$ is strongly negatively associated with both the mean and the variance of non-zero trade flows:
\[
\rho\!\left(\beta_a,\;\bar{Y}^{(\mathrm{nz})}_a\right)=-0.991,
\qquad
\rho\!\left(\beta_a,\;\mathrm{Var}(Y^{(\mathrm{nz})}_a)\right)=-0.947.
\]
\Cref{fig:beta_vs_mean_var_nz} visualizes these relationships (with the $x$-axes on a log scale for readability).

\begin{figure}[t]
  \centering
  \begin{subfigure}[t]{0.48\textwidth}
    \centering
    \includegraphics[width=\linewidth]{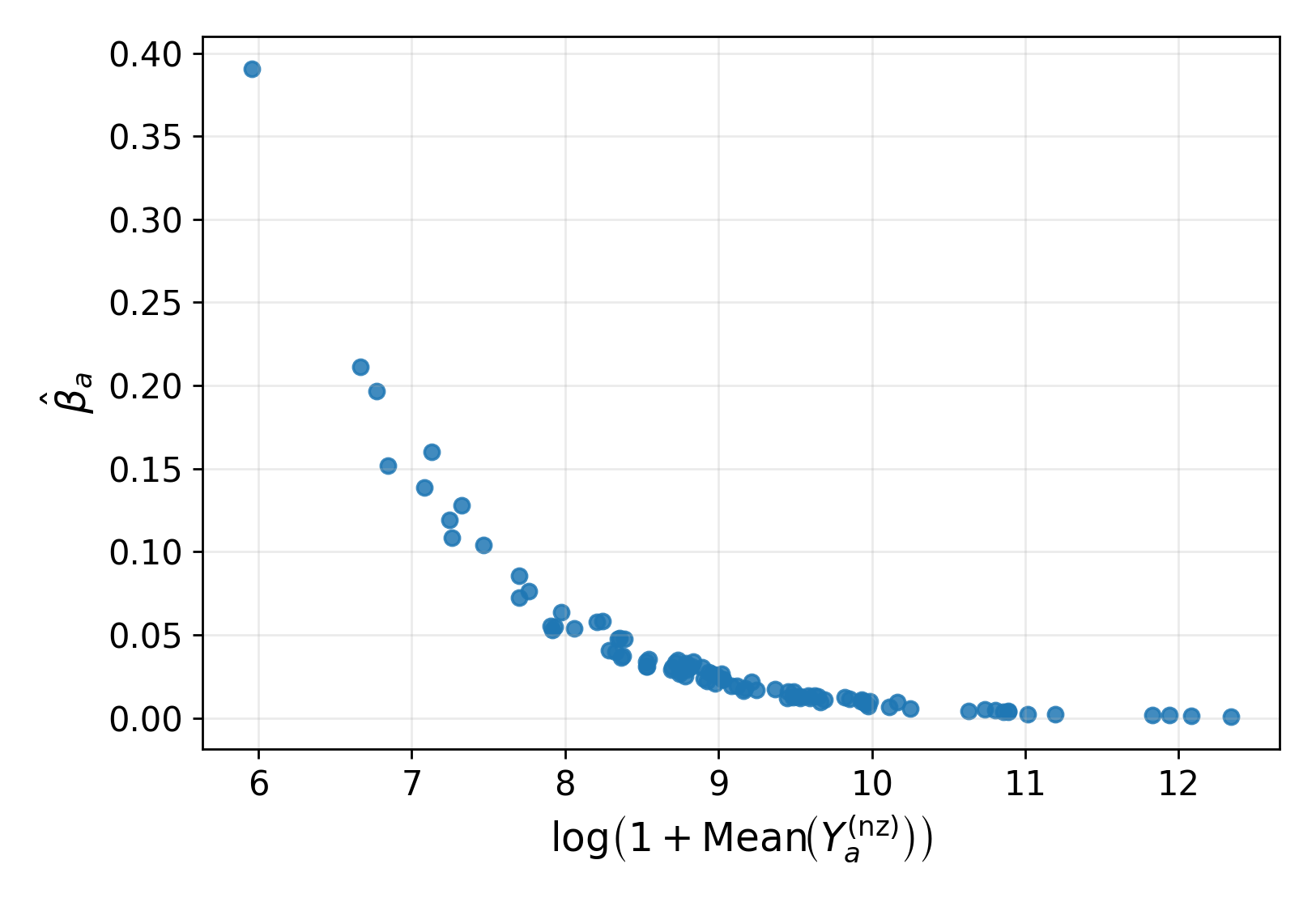}
    \caption{$\beta_a$ vs.\ $\log(\bar{Y}^{(\mathrm{nz})}_a)$}
  \end{subfigure}\hfill
  \begin{subfigure}[t]{0.48\textwidth}
    \centering
    \includegraphics[width=\linewidth]{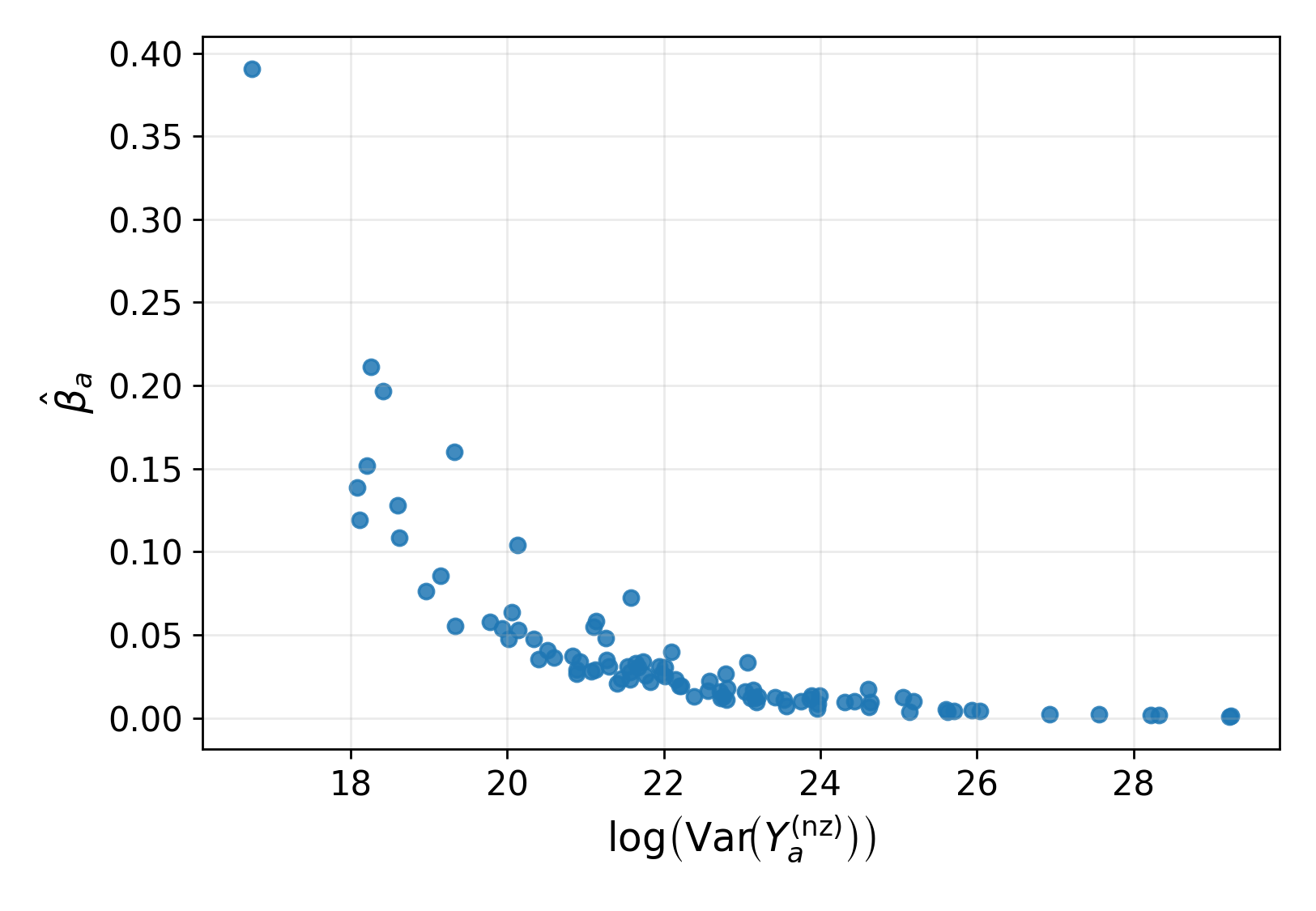}
    \caption{$\beta_a$ vs.\ $\log(\mathrm{Var}(Y^{(\mathrm{nz})}_a))$}
  \end{subfigure}
  \caption{Association between $\beta_a$ and the magnitude in panel (a), and dispersion in panel (b), of non-zero trade flows for each product $a$. The relationship is predominantly monotone and approximately decreasing: products with larger non-zero mean or variance tend to have smaller $\hat{\beta}_a$.}
  \label{fig:beta_vs_mean_var_nz}
\end{figure}

The product-specific estimates $\hat\beta_a$ are listed in \Cref{tab:beta_a_fixed_profile}.

\begingroup
\scriptsize           
\setlength{\tabcolsep}{3pt}        
\renewcommand{\arraystretch}{0.85} 
\begin{longtable}{r p{0.28\textwidth} r @{\hspace{1em}} r p{0.28\textwidth} r}
\caption{Estimated $\hat\beta_a$ by HS product group (indices 1--97; HS 77 is a placeholder).}
\label{tab:beta_a_fixed_profile}\\
\hline
\textbf{Idx} & \textbf{Product group} & \textbf{$\hat\beta_a$} &
\textbf{Idx} & \textbf{Product group} & \textbf{$\hat\beta_a$} \\
\hline
\endfirsthead
\hline
\textbf{Idx} & \textbf{Product group} & \textbf{$\hat\beta_a$} &
\textbf{Idx} & \textbf{Product group} & \textbf{$\hat\beta_a$} \\
\hline
\endhead
\hline
\multicolumn{6}{r}{\emph{Continued on next page}}\\
\endfoot
\hline
\endlastfoot
1  & live animals                 & 0.0308 & 50 & silk                          & 0.1082 \\
2  & meat products                & 0.0058 & 51 & wool/animal hair              & 0.0373 \\
3  & fish and seafood             & 0.0110 & 52 & cotton                        & 0.0194 \\
4  & dairy products               & 0.0128 & 53 & other veg.\ textiles          & 0.1386 \\
5  & other animal products        & 0.0765 & 54 & man-made filaments            & 0.0229 \\
6  & live plants and flowers      & 0.0311 & 55 & man-made staple fibers        & 0.0288 \\
7  & vegetables and roots         & 0.0167 & 56 & wadding/felt/nonwovens        & 0.0580 \\
8  & fruits and nuts              & 0.0120 & 57 & carpets/floor coverings       & 0.0636 \\
9  & coffee and spices            & 0.0281 & 58 & special woven/lace/embroidery & 0.0855 \\
10 & cereals                      & 0.0071 & 59 & coated/laminated fabrics      & 0.0473 \\
11 & milling products             & 0.0538 & 60 & knitted/crocheted fabrics     & 0.0252 \\
12 & oil seeds and grains         & 0.0124 & 61 & knit apparel/accessories      & 0.0097 \\
13 & resins and saps              & 0.1039 & 62 & woven apparel/accessories     & 0.0101 \\
14 & plaiting materials           & 0.3906 & 63 & other made-up textiles/rags   & 0.0334 \\
15 & animal/vegetable fats        & 0.0121 & 64 & footwear/parts                & 0.0132 \\
16 & prepared meat/seafood        & 0.0193 & 65 & headgear/parts                & 0.1601 \\
17 & sugars                       & 0.0292 & 66 & umbrellas/canes/parts         & 0.1964 \\
18 & cocoa products               & 0.0238 & 67 & feathers/artifc flowers/hair  & 0.0722 \\
19 & prepared cereal              & 0.0232 & 68 & stone/plaster/cement articles & 0.0350 \\
20 & prepared vegetables/fruits   & 0.0254 & 69 & ceramic products              & 0.0329 \\
21 & misc.\ food preparations     & 0.0276 & 70 & glass/glassware               & 0.0267 \\
22 & beverages and alcohol        & 0.0157 & 71 & pearls/precious metals/coin   & 0.0023 \\
23 & food industry waste/feed     & 0.0128 & 72 & iron/steel                    & 0.0038 \\
24 & tobacco products             & 0.0208 & 73 & iron/steel articles           & 0.0093 \\
25 & salt/sulfur/stone            & 0.0307 & 74 & copper/articles               & 0.0085 \\
26 & ores/slag/ash                & 0.0023 & 75 & nickel/articles               & 0.0164 \\
27 & mineral fuels/oils           & 0.0007 & 76 & aluminum/articles             & 0.0102 \\
28 & inorganic chemicals          & 0.0128 & 77 & reserved (placeholder)        & ---    \\
29 & organic chemicals            & 0.0038 & 78 & lead/articles                 & 0.0551 \\
30 & pharma products              & 0.0039 & 79 & zinc/articles                 & 0.0365 \\
31 & fertilizers                  & 0.0122 & 80 & tin/articles                  & 0.0529 \\
32 & dyes/pigments/paints         & 0.0219 & 81 & other base metals/cermets     & 0.0352 \\
33 & essential oils/perfume       & 0.0160 & 82 & tools/cutlery/base metal      & 0.0337 \\
34 & soap/detergents/waxes        & 0.0311 & 83 & misc.\ base-metal articles    & 0.0303 \\
35 & glues/enzymes                & 0.0476 & 84 & machinery/mech.\ appliances   & 0.0017 \\
36 & explosives/matches           & 0.1193 & 85 & electrical machinery/electronics & 0.0013 \\
37 & photo/film goods             & 0.0340 & 86 & railway/tramway equipment     & 0.0179 \\
38 & misc.\ chemical products     & 0.0109 & 87 & vehicles/parts                & 0.0016 \\
39 & plastics/articles            & 0.0051 & 88 & aircraft/spacecraft/parts     & 0.0043 \\
40 & rubber/articles              & 0.0124 & 89 & ships/boats/structures        & 0.0064 \\
41 & raw hides/skins/leather      & 0.0266 & 90 & optical/medical instruments   & 0.0047 \\
42 & leather goods/travel goods   & 0.0265 & 91 & clocks/watches/parts          & 0.0220 \\
43 & furskins/artificial fur      & 0.0549 & 92 & musical instruments/parts     & 0.1279 \\
44 & wood/articles                & 0.0136 & 93 & arms/ammunition/parts         & 0.0409 \\
45 & cork/articles                & 0.1519 & 94 & furniture/lighting/prefab     & 0.0100 \\
46 & plaited goods/basketware     & 0.2110 & 95 & toys/games/sporting goods     & 0.0173 \\
47 & wood pulp/paper scrap        & 0.0098 & 96 & misc.\ manufactured articles  & 0.0580 \\
48 & paper/paperboard articles    & 0.0114 & 97 & art/collectors/antiques       & 0.0399 \\
49 & printed books/media          & 0.0478 &    &                               &        \\
\end{longtable}
\endgroup

\subsection{Data-definition artifacts}
\label{app:data_artifacts}

While most components correspond to economically meaningful trade patterns, a small number primarily reflect reporting conventions or definitional breaks. \Cref{fig:usaircraft_supp} highlights discontinuities in the U.S.\ aircraft series---including near-zero values in 2004--2005, 2009--2015, and after 2022---that are more plausibly attributable to aircraft-specific reporting revisions and the post-2009 consolidation into HS880000 noted by the U.S.\ Census Bureau, together with subsequent HS concordance changes, than to genuine trade fluctuations \citep{USCensus2009Aircraft,ZhuYamanoCimper2011}. We treat such components as useful data diagnostics: the factorization separates measurement artifacts from substantive trade structure.

\begin{figure}[!htbp]
    \centering
    \includegraphics[width=0.6\linewidth]{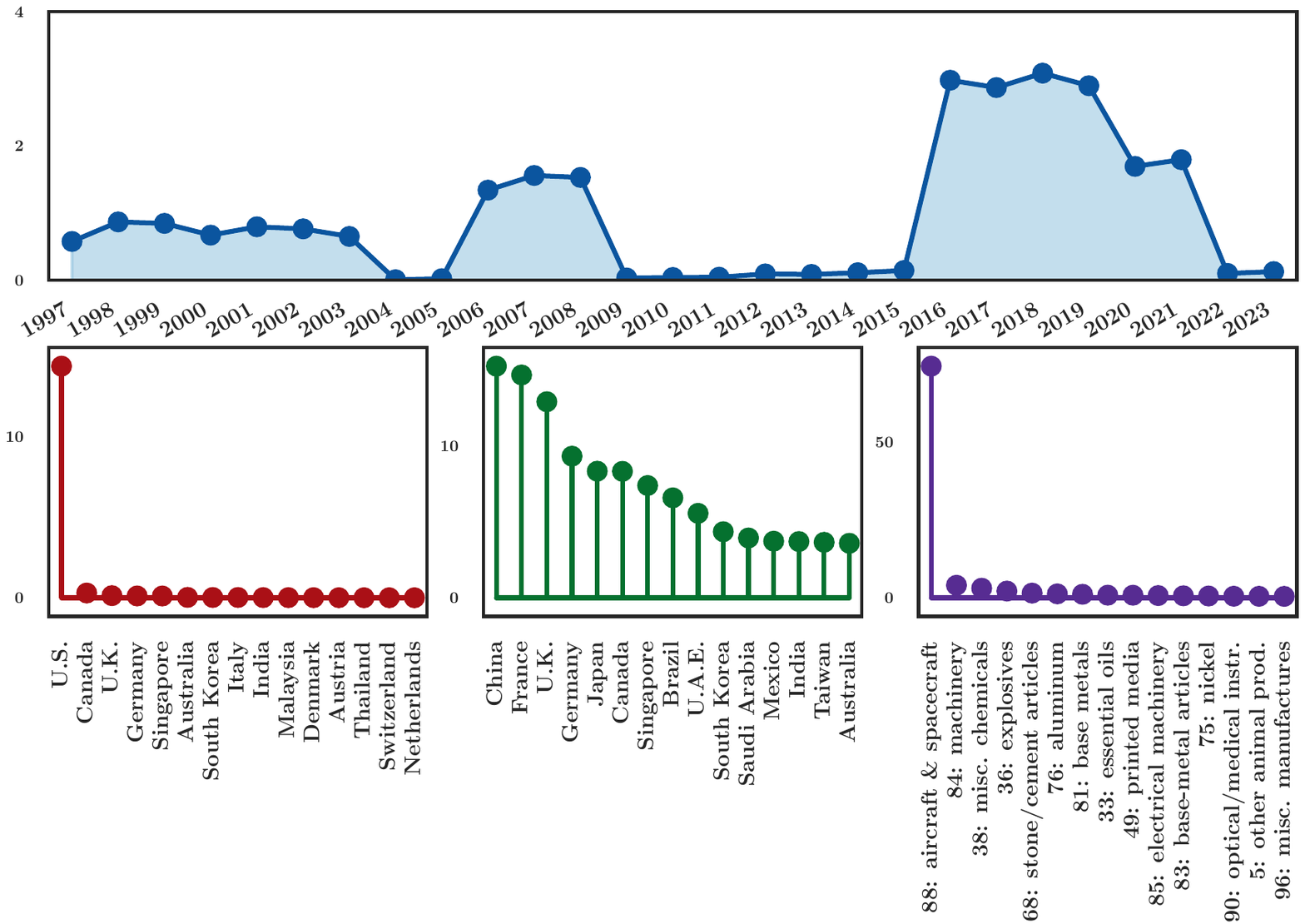}
    \caption{\textbf{Data-definition component: U.S.\ aircraft reporting and aggregation change.} This component reflects reporting conventions or definitional breaks rather than stable trade structure.}
    \label{fig:usaircraft_supp}
\end{figure}

\subsection{Product uniqueness: detailed list}
\label{app:product_uniqueness}

Components in the bottom quintile of product-mode entropy $H^P_k$ are often dominated by a single HS2 category. The most frequently dominant categories include HS71 (pearls/precious metals/coins), HS91 (clocks/watches/parts), HS88 (aircraft/spacecraft/parts), HS89 (ships/ boats/structures), HS97 (art/collectors/antiques), and HS29--30 (organic chemicals and pharmaceuticals), as well as resource- or geography-constrained agri-food categories such as HS09 (coffee/spices), HS12 (oil seeds), HS15 (animal/vegetable fats), HS18 (cocoa products), HS22 (beverages), and HS45 (cork). These patterns are consistent with well-established findings that certain product classes are exported by a small set of countries, either due to specialized capabilities (e.g., aircraft, pharmaceuticals, precision goods) or concentrated endowments (e.g., cocoa, coffee, oil seeds, cork), with additional concentration reflecting trading hubs (e.g., precious metals and art). See \citet{HidalgoHausmann2009} and \citet{Rauch1999}.

\subsection{Binned entropy summary}
\label{app:entropy_binned}

\Cref{fig:sprod_smin_bins_supp} presents a binned summary of the product-uniformity versus partner-concentration trade-off shown in Figure~\ref{fig:entropy_example} of the main text.

\begin{figure}[!htbp]
  \centering
  \includegraphics[width=0.5\linewidth]{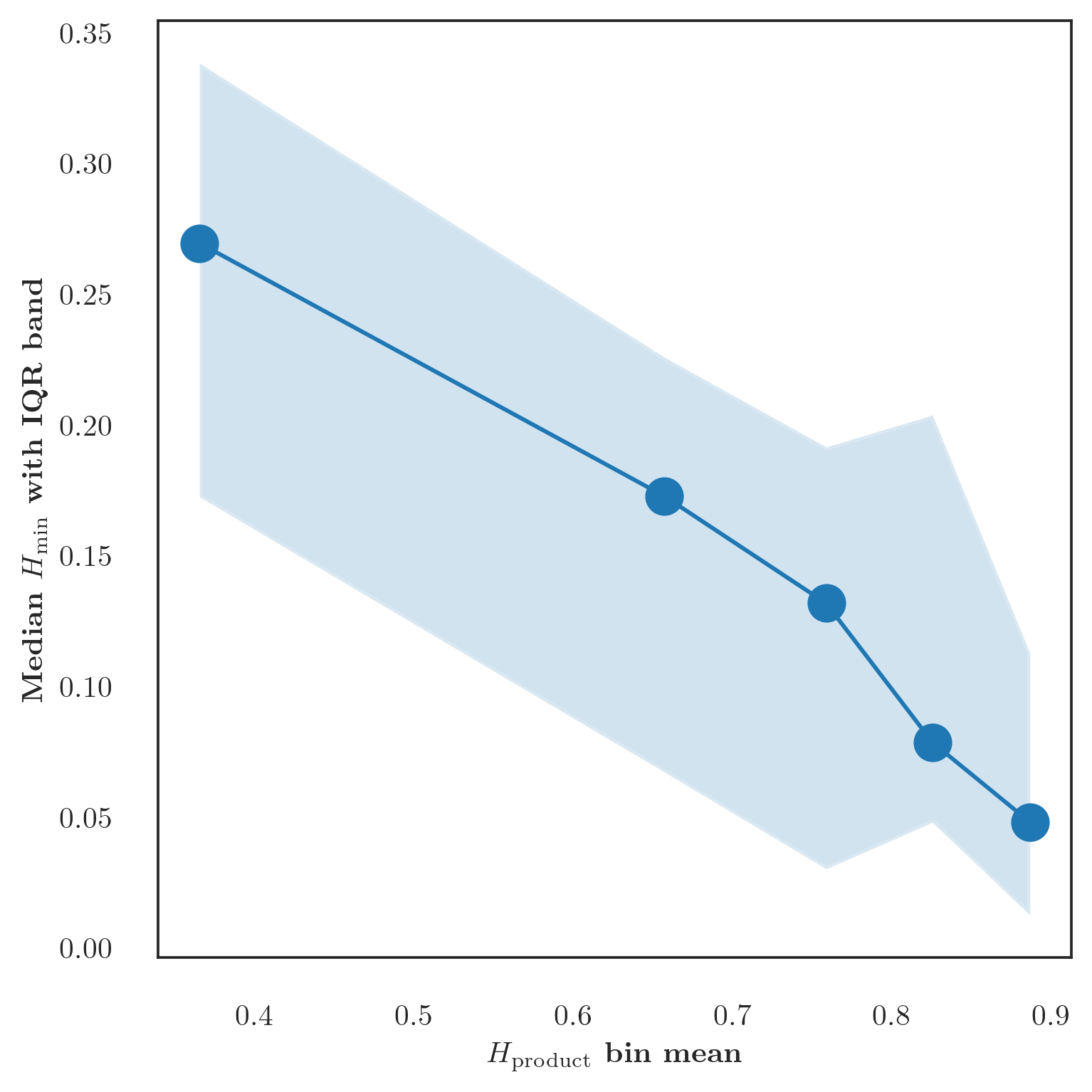}
  \caption{Binned median of $H_{\min,k}$ with interquartile-range band, grouped by product-mode entropy $H^{P}_k$.}
  \label{fig:sprod_smin_bins_supp}
\end{figure}

\subsection{Grubel--Lloyd index in latent product space}
\label{sec:supp_gl}

For each ordered pair $\{i,j\}$ where both directions of trade are observed, define the directed channel profile
\begin{equation}
\pi_{ij}(k) \;=\; \frac{\tilde\theta^{(1)}_{ik}\tilde\theta^{(2)}_{jk}}{\sum_{k'}\tilde\theta^{(1)}_{ik'}\tilde\theta^{(2)}_{jk'}},
\label{eq:supp_gl_pi}
\end{equation}
the share of bilateral trade from $i$ to $j$ flowing through latent channel $k$. Project through the product factor to obtain the implied product distribution
\begin{equation}
q_{ij}(a) \;=\; \sum_{k=1}^{125} \pi_{ij}(k)\,\tilde\theta^{(3)}_{ak},
\label{eq:supp_gl_q}
\end{equation}
and compute the Grubel--Lloyd index
\begin{equation}
\mathrm{GL}_{ij} \;=\; 1 \;-\; \frac{\sum_a |q_{ij}(a) - q_{ji}(a)|}{\sum_a \big(q_{ij}(a) + q_{ji}(a)\big)} \;=\; 1 - \mathrm{TV}\big(q_{ij},\,q_{ji}\big),
\label{eq:supp_gl}
\end{equation}
where $\mathrm{TV}$ denotes total-variation distance. Since $q_{ij}$ and $q_{ji}$ are probability distributions, the denominator equals~$2$ and the index is bounded in $[0,1]$.

Figure~\ref{tab:gl} in the main text reports the 25 most symmetric and most asymmetric pairs among those above the 95th percentile of bilateral volume. Restricting to high-volume pairs avoids trivially symmetric scores produced by very low-trade pairs (e.g., neighbouring small economies whose only bilateral trade consists of a few border goods). Figure~\ref{fig:4pairs} in the main text shows four representative pairs --- the most- and least-symmetric, together with two intermediate cases. The two extremes recovered by the index correspond to the two regimes described in the main text: pairs of similar economies that trade comparable product mixes in both directions (high GL), and pairs of structurally different economies whose trade is specialized along factor-endowment lines (low GL) \citep{helpman_krugman1985,feenstra2016advanced}. The decomposition recovers both patterns simultaneously from a single unsupervised model with no country labels, income groups, or product hierarchies supplied as input.

\subsection{China to Hong Kong textiles component}
\label{app:china_hk}

\Cref{fig:CHN-HK} shows a China$\to$Hong Kong cotton/textiles/knitwear component whose time factor declines with a pronounced drop around 2005. Hong Kong has long served as an entrep\^{o}t for Mainland trade \citep{hktid_mainland_entrepot}, and a large share of China's exports historically passed through Hong Kong intermediaries in light manufactures \citep{feenstra_hanson_2001}. The decline coincides with the expiry of the WTO Agreement on Textiles and Clothing (ATC) on 1 January 2005 \citep{wto_atc_legal,ers_mfa_2006}, which removed quota-based incentives to route textile shipments through re-export hubs and coincided with broader shifts toward direct exporting.

\begin{figure}[!htbp]
    \centering
    \includegraphics[width=0.55\linewidth]{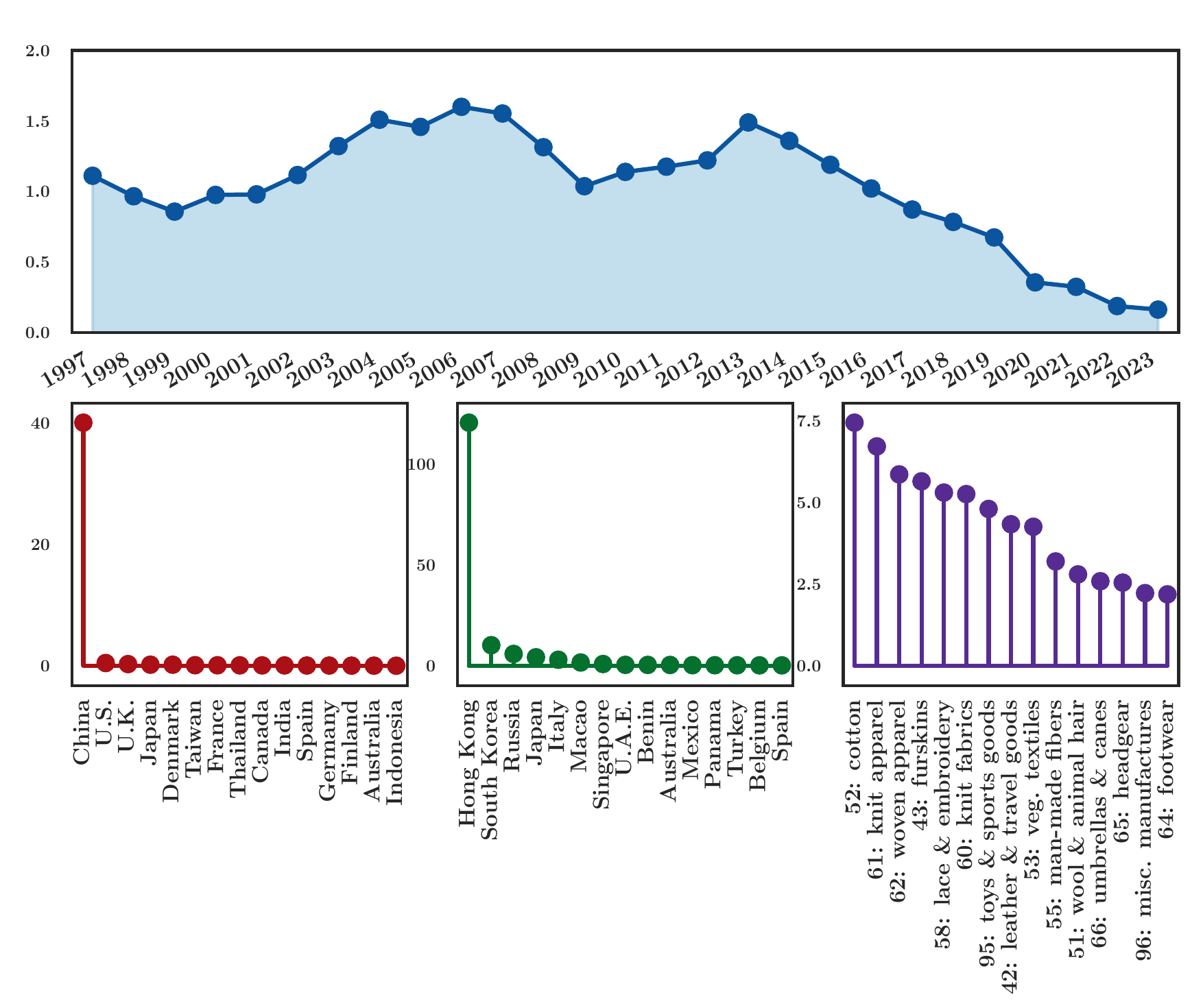}
    \caption{\textbf{China$\to$Hong Kong textiles/knitwear component (declining).} The time factor declines over the sample with a marked drop around 2005, consistent with reduced re-export routing following the end of the MFA/ATC quota regime \citep{feenstra_hanson_2001,hktid_mainland_entrepot,wto_atc_legal,ers_mfa_2006}.}
    \label{fig:CHN-HK}
\end{figure}

\end{document}